\documentclass[a4paper]{amsart}
\usepackage{amscd}
\usepackage{amsmath}
\usepackage{amssymb}
\usepackage{amsthm}
\usepackage{bbm}
\usepackage{stmaryrd}

\usepackage[T1]{fontenc}

\newif\ifpdf
\ifx\pdfoutput\undefined
   \pdffalse        
\else
   \pdfoutput=1     
   \pdftrue
\fi

\ifpdf
   \usepackage[pdftex]{graphicx}
\else
   \usepackage{graphicx}
\fi

\numberwithin{equation}{section} \swapnumbers

\newtheorem{satz}{Satz}[section]

\newtheorem{theorem}[satz]{Theorem}
\newtheorem{proposition}[satz]{Proposition}
\newtheorem{corollary}[satz]{Corollary}
\newtheorem{lemma}[satz]{Lemma}
\newtheorem{assumption}[satz]{Assumption}

\newtheorem{definition}[satz]{Definition}

\newtheorem{remark}[satz]{Remark}

\newtheorem{example}[satz]{Example}

\newcommand{\bbr}{\mathbb{R}}
\newcommand{\bbe}{\mathbb{E}}
\newcommand{\bbn}{\mathbb{N}}
\newcommand{\bbp}{\mathbb{P}}
\newcommand{\bbq}{\mathbb{Q}}

\newcommand{\bbc}{\mathbb{C}}

\newcommand{\calb}{\mathcal{B}}
\newcommand{\cald}{\mathcal{D}}
\newcommand{\cale}{\mathcal{E}}
\newcommand{\calf}{\mathcal{F}}

\newcommand{\call}{\mathcal{L}}
\newcommand{\calm}{\mathcal{M}}
\newcommand{\calo}{\mathcal{O}}
\newcommand{\calp}{\mathcal{P}}
\newcommand{\calv}{\mathcal{V}}
\newcommand{\calx}{\mathcal{X}}
\newcommand{\caly}{\mathcal{Y}}

\begin{document}

\title[Real-world forward rate dynamics with affine realizations]{Real-world forward rate dynamics with affine realizations}
\author{Eckhard Platen \and Stefan Tappe}
\address{University of Technology Sydney, School of Mathematical Sciences and Finance Discipline Group, PO Box 123, Broadway, NSW 2007, Australia}
\email{eckhard.platen@uts.edu.au}
\address{Leibniz Universit\"{a}t Hannover, Institut f\"{u}r Mathematische Stochastik, Welfengarten 1, 30167 Hannover, Germany}
\email{tappe@stochastik.uni-hannover.de}
\begin{abstract}
We investigate the existence of affine realizations for L\'{e}vy driven interest rate term structure models under the real-world probability measure, which so far has only been studied under an assumed risk-neutral probability measure. For models driven by Wiener processes, all results obtained under the risk-neutral approach concerning the existence of affine realizations are transferred to the general case. A similar result holds true for models driven by compound Poisson processes with finite jump size distributions. However, in the presence of jumps with infinite activity we obtain severe restrictions on the structure of the market price of risk; typically, it must even be constant.
\end{abstract}
\maketitle

\section{Introduction}

The purpose of this paper is to investigate when a HJM (Heath-Jarrow-Morton) interest rate term structure model
\begin{align}\label{HJMM}
\left\{
\begin{array}{rcl}
dr_t & = & \big( \frac{d}{d\xi} r_t + \alpha(r_t,Y_t) \big) dt +
\sigma(r_{t})dW_t + \gamma(r_{t-})dX_t
\medskip
\\ r_0 & = & h_0 \medskip
\\ Y_0 & = & y_0
\end{array}
\right.
\end{align}
in the framework of the Benchmark Approach (see \cite{Platen}) admits an affine realization. Here $W$ is a $\bbr^d$-valued Wiener process and $X$ is a $\bbr^n$-valued pure jump L\'{e}vy process $X$ with components having the canonical representations $X^k = x * \mu^{X^k}$ for $k=1,\ldots,m$ and $X^k = x * (\mu^{X^k} - \nu^k)$ for $k=m+1,\ldots,n$, where $\nu^k$ denotes the respective compensator. Under risk-neutral pricing, we refer to \cite{HJM} for the classical HJM model driven by Wiener processes, and, e.g., to \cite{Eberlein_J}--\cite{Eberlein-Raible} for HJM models driven by L\'{e}vy processes. We study the term structure equation (\ref{HJMM}) under the real-world probability measure, and with Musiela parametrization (see \cite{Musiela}), which gives rise to a stochastic partial differential equation (SPDE) in the spirit of \cite{P-Z-book} on some appropriate Hilbert space $H$, whence we will refer to (\ref{HJMM}) as HJMM (Heath-Jarrow-Morton-Musiela) equation. The risk-neutral HJMM equation has been investigated, e.g., in \cite{fillnm, Filipovic-Tappe, Barski, P-Z-paper, Marinelli}. The process $Y$ in (\ref{HJMM}) is an external state process on some state space $\caly$, which appears in the drift term (\ref{alpha-HJM}) below. In order to ensure the absence of arbitrage in the bond market
\begin{align*}
P_t(T) = \exp \bigg( -\int_0^{T-t} r_t(\xi) d\xi \bigg)
\end{align*}
within the framework of the Benchmark Approach, benchmarked bond prices have to be local martingales, which is ensured by choosing a drift term of the form
\begin{equation}\label{alpha-HJM}
\begin{aligned}
\alpha(h,y) &= -\sum_{k=1}^d \big( \sigma^k(h) \Sigma^k(h) - \Theta^k(y) \sigma^k(h) \big) 
\\ &\quad - \sum_{k=1}^m \gamma^k(h) \int_{\mathbb{R}} x \Phi^k(y,x) e^{x \Gamma^k(h)} F^k(dx)
\\ &\quad - \sum_{k=m+1}^n \gamma^k(h) \int_{\mathbb{R}} x \big( \Phi^k(y,x) e^{x \Gamma^k(h)} - 1 \big) F^k(dx).
\end{aligned}
\end{equation}
We refer to Section \ref{sec-bond-benchmark} for a review of the Benchmark Approach, and to Section \ref{sec-HJM-benchmark} for the derivation of the drift condition (\ref{alpha-HJM}). Here we use the notations $\Sigma(h) = -\int_0^{\bullet} \sigma(h)(\xi) d\xi$ and $\Gamma(h) = -\int_0^{\bullet} \gamma(h)(\xi) d\xi$, and the $F^k$ are the L\'{e}vy measures. Furthermore, $(\theta,\psi) = (\Theta(Y),\Psi(Y))$ denotes a pair of market prices of risk, and we have set $\Phi(Y) = 1 - \Psi(Y)$. We call $(\theta,\psi)$ a pair of market prices of risk, because for each $T \in \bbr_+$ the dynamics of the bond prices are of the form
\begin{align}\label{bond-dyn-MPR}
P(T) &= P_{0}(T) \, \cale \big( ( R + a(T) ) \cdot \lambda + b(T) \cdot W + c(T) * (\mu^X - \nu) \big),
\end{align}
where $R$ denotes the short rate and $(\theta,\psi)$ is a solution of the equation
\begin{align}\label{MPR-solution}
a(T) = \langle b(T), \theta \rangle_{\bbr^d} + \langle c(T), \psi \rangle_{L^2(F)}.
\end{align}
Furthermore, the strictly positive supermartingale
\begin{align}\label{Z-candidate-intro}
Z = \cale \big( -\theta \cdot W - \psi * (\mu^X - \nu) \big)
\end{align}
defines a candidate for the density process of an equivalent local martingale measure, and it provides an equivalent local martingale measure if and only if
\begin{align}\label{uniformly}
\text{$Z$ is a uniformly integrable martingale with $\bbp(Z_{\infty} > 0) = 1$.}
\end{align}
The existence of an affine realization for the HJMM equation (\ref{HJMM}) ensures larger analytical tractability of the model, and there exists a well established literature on affine realizations for term structure models under the classical risk-neutral approach. We refer, e.g., to \cite{Bj_Sv, Bj_La, Filipovic, Tappe-Wiener} for Wiener process driven models, and to \cite{Tappe-Levy} for L\'{e}vy process driven models. In all these references, the main idea of an affine realization is that for every initial curve $h_0$ there exists a finite dimensional submanifold, on which the solution process $r$ stays. Compared to this risk-neutral definition of an affine realization, in our framework we demand that for every initial curve $h_0$ there exists a finite dimensional submanifold such that for each starting point $y_0$ of the state process $Y$ the solution $r$ to the HJMM equation (\ref{HJMM}) stays on this submanifold, which means that we are free to specify the market price of risk.

Our first goal of this paper is to derive a criterion which refers to the risk-neutral case, for which the aforementioned literature is available. Namely, our first main result (see Theorem~\ref{thm-HJMM-affine-real}) states that the HJMM equation (\ref{HJMM}) has an affine realization if and only if the following two conditions are satisfied:
\begin{enumerate}
\item[(i)] The risk-neutral HJMM equation has an affine realization.

\item[(ii)] We have $\dim U_{\Psi,\gamma} < \infty$.
\end{enumerate}
Here the risk-neutral HJMM equation corresponds to $(\Theta,\Psi) = 0$, but in our framework we do not assume the existence of an equivalent local martingale measure, and the subspace $U_{\Psi,\gamma} \subset H$ is defined as
\begin{align}\label{def-U-Psi-gamma}
U_{\Psi,\gamma} := \Big\langle \sum_{k=1}^n \int_{\bbr} \Psi^k(y,x) e^{x \Gamma^k(h)} F^k(dx) : h \in H \text{ and } y \in \caly \Big\rangle.
\end{align}
As point (i) has intensively been studied in the literature, our next goal is to have a closer look at condition (ii) to find equivalent conditions, which are easier to check. Our subsequent results (see Proposition~\ref{prop-fin-dim-suff} and Theorems~\ref{thm-dim-U-Psi}, \ref{thm-convex}) show that, under suitable assumptions, condition (ii) is equivalent to the following two conditions:
\begin{enumerate}
\item[(a)] We have $\dim U_{\Psi^k} < \infty$ for $k=1,\ldots,n$.

\item[(b)] We have $\dim U_{\gamma^k} < \infty$ for $k=1,\ldots,n$.
\end{enumerate}
Here the subspaces $U_{\Psi^k} \subset L^2(F^k)$ and $U_{\gamma^k} \subset L^2(F^k;H)$ are defined as
\begin{align}\label{def-U-Psi}
U_{\Psi^k} &:= \langle x \mapsto \Psi^k(y,x) : y \in \caly \rangle, \quad k=1,\ldots,n,
\\ \label{def-U-gamma} U_{\gamma^k} &:= \langle x \mapsto e^{x \Gamma^k(h)} : h \in H \rangle, \quad k=1,\ldots,n.
\end{align}
Conditions (ii) and (a) lead to the consequence that in the presence of jumps with infinite activity the market price of risk is subject to severe restrictions. We will see that, as a further consequence, the market price of risk must typically even be constant in this case; see Proposition~\ref{prop-cumulant} for such a result which follows from (ii), and Proposition~\ref{prop-Fc-Fd} for such a result which follows from (a).

Furthermore, Theorem~\ref{thm-convex} even shows that, under suitable assumptions, condition (ii) implies that the volatility $\gamma$ is constant. 
Therefore, we arrive at the conclusion that, under suitable assumptions, conditions (a) and (b) are equivalent to the following two conditions:
\begin{enumerate}
\item[(a')] The L\'{e}vy processes $X^k$, $k=1,\ldots,n$ are compound Poisson processes with finite jump size distributions.

\item[(b')] The volatilities $\gamma^k$, $k=1,\ldots,n$ are constant.
\end{enumerate}
After this outline regarding conditions which are equivalent to (ii), let us proceed with interpretations of condition (ii). Note that the subspace $U_{\Psi,\gamma}$ only depends on $\Psi$ and $\gamma$. Thus, for purely Wiener process driven models without jumps, the existence of an affine realization is equivalent to the existence of an affine realization in the risk-neutral case, whereas for interest rate models with jumps we need the additional condition that the subspace $U_{\Psi,\gamma}$ is finite dimensional. In Remarks~\ref{remark-geometric} and \ref{remark-geometric-2} we will provide geometric interpretations, which we shall summarize here:
\begin{itemize}
\item The first interpretation is a differential geometric interpretation:
\begin{itemize}
\item In the Wiener process driven case, condition (i) implies that for a given submanifold the required tangential conditions, which we need for stochastic invariance, are already fulfilled for each choice of $y_0 \in \caly$.

\item In contrast, if the model has jumps, then these tangential conditions are fulfilled if and only if the subspace $U_{\Psi,\gamma}$ is finite dimensional.
\end{itemize}

\item The second interpretation concerns measure changes. Every choice of $y_0 \in \caly$ gives rise to the candidate (\ref{Z-candidate-intro}) for the density process of an equivalent local martingale measure. If condition (\ref{uniformly}) is fulfilled, then the following statements are true:
\begin{itemize}
\item For Wiener process driven models without jumps, the drift term under the new probability measure coincides with the classical HJM drift condition. Therefore, changing $y_0 \in \caly$ leads to dynamics after an equivalent measure change, which does not affect stochastic invariance of a given submanifold.

\item Otherwise, in the presence of jumps, the drift term under the new probability measure does not coincide with the HJM drift term of a risk-neutral model. Therefore, changing $y_0 \in \caly$ can usually not be associated to an equivalent measure change, and hence, stochastic invariance is not preserved.
\end{itemize}
\end{itemize}
The remainder of this paper is organized as follows. In Sections \ref{sec-bond-benchmark} and \ref{sec-HJM-benchmark} we review basic ideas and concepts concerning bond market models under the Benchmark Approach. In Section~\ref{sec-affine-real-SPDE} we provide results on invariant foliations and
on affine realizations for general SPDEs driven by L\'evy processes. In Section \ref{sec-affine-real-general} we deal with affine realizations for the HJMM equation and present the indicated result regarding the subspace $U_{\Psi,\gamma}$. Then, as special cases, in Section \ref{sec-Wiener} we investigate the Wiener process driven HJMM equation, and in Section \ref{sec-Levy-cumulant} the L\'{e}vy process driven HJMM equation, where the drift term can be described by the cumulant generating function of the L\'{e}vy process. In Section \ref{sec-suff-prod} we show that conditions (a) and (b) imply condition (ii). In Sections \ref{sec-nec-Psi} and \ref{sec-nec-gamma} we deal with the converse implication of this result. Section \ref{sec-conclusion} concludes. For convenience of the reader, Appendices~\ref{app-analytic}--\ref{app-lin-ind} provide auxiliary results which we need in this paper.

\section{Bond market models under the Benchmark Approach}\label{sec-bond-benchmark}

In this section, we provide a review of the basic ideas and concepts of the Benchmark Approach, which we require in this paper for bond market models. The upcoming definitions and results are well-known and can be found in \cite{Platen}, but we provide them (with proofs) in order to keep our presentation self-contained and to introduce notation, which we will need in further sections.

From now on, let $(\Omega,\mathcal{F},(\mathcal{F}_t)_{t \geq 0},\mathbb{P})$ be a filtered probability space with right-continuous filtration. In the sequel, we will use the notation from \cite{Jacod-Shiryaev}; in particular $\delta \cdot P$ denotes the stochastic integral of a locally bounded, predictable process $\delta$ with respect to a semimartingale $P$. For each $T \in \bbr_+$ let $P(T) = (P_t(T))_{t \in [0,T]}$ be the price process of a zero coupon bond with maturity $T$, which we assume to be a nonnegative semimartingale with $P_T(T) = 1$. We recall some basic concepts from the theory of asset pricing.

\begin{definition}
For each $n \in \bbn$ and all $0 \leq T_1 < \ldots < T_n < \infty$ we call a vector $\delta = (\delta^{T_1},\ldots,\delta^{T_n})$ consisting of locally bounded, predictable processes $\delta^{T_k} = (\delta_t^{T_k})_{t \in [0,T_k]}$ a \emph{strategy}.
\end{definition}

\begin{definition}\label{def-portfolio}
For a strategy $\delta = (\delta^{T_1}, \ldots, \delta^{T_n})$ we define the \emph{portfolio} $S^{\delta} = (S_t^{\delta})_{t \in [0,T_1]}$ as the vector inner product $S^{\delta} := \delta P$, where $P = (P(T_1),\ldots,P(T_n))$.
\end{definition}

\begin{definition}
A strategy $\delta = (\delta^{T_1}, \ldots, \delta^{T_n})$ and the corresponding portfolio $S^{\delta}$ are called \emph{self-financing}, if we have
\begin{align*}
S^{\delta} = S_0^{\delta} + \delta \cdot P,
\end{align*}
where $P = (P(T_1),\ldots,P(T_n))$, and $\delta \cdot P$ denotes the vector It\^{o} integral.
\end{definition}

\begin{definition}\label{def-arbitrage}
A nonnegative self-financing portfolio $S^{\delta}$ is called an \emph{arbitrage portfolio}, if $S_0^{\delta} = 0$ and there is a stopping time $\tau \leq T_1$ such that $\mathbb{P}(S_{\tau}^{\delta} > 0) > 0$.
\end{definition}

Note that Definition \ref{def-arbitrage} is a rather weak notion of arbitrage, since we only consider nonnegative portfolios.

\begin{definition}\label{def-gop}
A strictly positive portfolio process $S^{\delta_*} = (S_t^{\delta_*})_{t \in \bbr_+}$ is called a \emph{growth optimal portfolio}, if for each nonnegative self-financing portfolio $S^{\delta}$ the \emph{benchmarked portfolio} $\hat{S}^{\delta} = (\hat{S}_t^{\delta})_{t \in [0,T_1]}$ defined as $\hat{S}^{\delta} := S^{\delta} / S^{\delta_*}$ is a local martingale.
\end{definition}

\begin{remark}
Let $S^{\delta_*}$ be a growth optimal portfolio and let $S^{\delta}$ be a nonnegative self-financing portfolio. Since $S^{\delta}$ is nonnegative and $S^{\delta_*}$ is positive, the benchmarked portfolio $\hat{S}^{\delta}$ is a nonnegative local martingale, and hence, a supermartingale.
\end{remark}

\begin{remark}\label{remark-max}
The name growth optimal portfolio comes from the result that $S^{\delta_*}$ is the portfolio which maximizes the expected log-utility; see, e.g., \cite{Platen} for further details.
\end{remark}

The importance of the growth optimal portfolio regarding arbitrage portfolios is demonstrated by the next result. For the sake of completeness, we provide its proof here.

\begin{proposition}\label{prop-no-arbitrage}
Suppose there is a growth optimal portfolio $S^{\delta_*}$. Then no arbitrage portfolio exists.
\end{proposition}

\begin{proof}
Let $S^{\delta}$ be a nonnegative self-financing portfolio such that $S_0^{\delta} = 0$. Furthermore, let $\tau \leq T_1$ be a stopping time. Since $\hat{S}^{\delta}$ is a nonnegative supermartingale, by Doob's optional sampling theorem we obtain
\begin{align*}
0 \leq \mathbb{E}[ \hat{S}_{\tau}^{\delta} ] \leq \mathbb{E}[ \hat{S}_{0}^{\delta} ] = 0,
\end{align*}
which yields 
\begin{align*}
\mathbb{E} [ S_{\tau}^{\delta} / S_{\tau}^{\delta_*} ] = \mathbb{E}[ \hat{S}_{\tau}^{\delta} ] = 0.
\end{align*}
Since $S^{\delta} / S^{\delta_*}$ is nonnegative, this gives us
\begin{align*}
\mathbb{P} ( S_{\tau}^{\delta} / S_{\tau}^{\delta_*} = 0 ) = 1,
\end{align*}
and, since $S^{\delta_*}$ is strictly positive, we arrive at $\mathbb{P}(S_{\tau}^{\delta} = 0) = 1$.
\end{proof}

Next, we recall how to perform real-world pricing under the Benchmark Approach.

\begin{definition}\label{def-option-pricing}
Let $S^{\delta_*}$ be a growth optimal portfolio, let $T \in \bbr_+$ be arbitrary, and let $H$ be a nonnegative $\calf_T$-measurable random variable such that $H / S_T^{\delta_*} \in \call^1(\bbp)$. We define the real-world price process $\pi^{\delta_*}(H) = (\pi_t^{\delta_*}(H))_{t \in [0,T]}$ via the real-world pricing formula
\begin{align}\label{real-world-formula}
\pi_t^{\delta_*}(H) := S_t^{\delta_*} \, \bbe_{\bbp} \bigg[ \frac{H}{S_T^{\delta_*}} \, \Big| \, \calf_t \bigg], \quad t \in [0,T].
\end{align}
\end{definition}

\begin{remark}
Note that Definition \ref{def-option-pricing} does not rely on the existence of a local martingale measure, and that the real-world price process $\pi^{\delta_*}(H)$ of a payoff profile $H$ is fair in the sense that the benchmarked real-world price process $\hat{\pi}^{\delta_*}(H) = \pi^{\delta_*}(H) / S^{\delta_*}$ is a martingale.
\end{remark}

\begin{remark}
If for a given contingent claim $H$ a self-financing portfolio $\pi^{\delta_*}(H)$ exists, satisfying the above real-world pricing formula (\ref{real-world-formula}), then this portfolio provides the least expensive hedge for $H$, see Prop.~3.3 in \cite{Du-Platen}. If one considers pricing under other pricing rules, e.g. formally applied risk-neutral pricing, then the corresponding benchmarked nonnegative, self-financing hedge portfolios are local martingales and, in general, more expensive. One can argue that in a competitive market the minimal possible price processes are the economically correct price processes, which underpins the special role of the real-world price processes of the Benchmark Approach.
\end{remark}

\begin{remark}
In Section \ref{sec-HJM-benchmark}, we will study bond markets of the form (\ref{bond-market}) with forward rates given by (\ref{HJM-model}). The form (\ref{HJM-drift}) of the drift term ensures that for each $T \in \bbr_+$ the benchmarked price processes for the zero coupon bond (\ref{bond-market}) is a local martingale. Under the real-world pricing formula
\begin{align*}
P_t(T) = S_t^{\delta_*} \, \bbe_{\bbp} \bigg[ \frac{1}{S_T^{\delta_*}} \, \Big| \, \calf_t \bigg]
\end{align*}
it is even a true martingale which represents the minimal possible bond price process for this payoff. Since any benchmarked self-financing portfolio is in our setting a local martingale, the structure of (\ref{HJM-model}) is very general and covers also other pricing rules that replicate the respective payoff in a self-financing manner by a nonnegative portfolio. These pricing rules do not need to be linked to any pricing measure; we refer to Section \ref{sec-HJM-benchmark} for further details.
\end{remark}

To link our approach more closely to the existing literature, let us illustrate the Benchmark Approach within the framework of num\'{e}raire pairs, as, for example, considered in \cite{Hunt}.

\begin{definition}
We introduce the following notions:
\begin{enumerate}
\item A pair $(N,\bbq)$ is called a \emph{num\'{e}raire pair} if $\bbq \sim \bbp$ is an equivalent probability measure, $N$ is a strictly positive semimartingale, and for each nonnegative self-financing portfolio $S^{\delta}$ the discounted portfolio $S^{\delta} / N$ is a $\bbq$-local martingale.

\item If $(N,\bbq)$ is a num\'{e}raire pair, then we call $N$ a \emph{num\'{e}raire}, and $\bbq$ a \emph{valuation measure}.
\end{enumerate}
\end{definition}

\begin{remark}
If $(S^{\delta_*},\bbp)$ is a num\'{e}raire pair, then $S^{\delta_*}$ is a growth optimal portfolio in the sense of Definition \ref{def-gop}. 
\end{remark}

In the classical framework, the bond market is called free of arbitrage, if a num\'{e}raire pair $(N,\bbq)$ exists. A typical choice for the num\'{e}raire $N$ is the savings account; see Section \ref{sec-HJM-benchmark}, where we investigate the situation in more detail for HJM models. 

The following result shows that absence of arbitrage in the classical sense implies absence of arbitrage in the spirit of the Benchmark Approach, and that, in this case, the price process of some payoff, as given by Definition \ref{def-option-pricing}, coincides with the classical risk-neutral pricing process. For the rest of this section, we suppose that $\calf = \bigvee_{t \in \bbr_+} \calf_t$.

\begin{proposition}\label{prop-num-1}
Let $(N,\bbq)$ be a num\'{e}raire pair, let $Z$ be the density process of $\bbq$ relative to $\bbp$, and let $S_0^{\delta_*}$ be a strictly positive $\calf_0$-measurable random variable. Then the following statements are true:
\begin{enumerate}
\item The process 
\begin{align}\label{GOP-num}
S^{\delta_*} = \frac{S_0^{\delta_*}}{N_0} \frac{N}{Z} 
\end{align}
is a growth optimal portfolio.

\item For each $T \in \bbr_+$ and every $\calf_T$-measurable random variable $H$ the following statements are true:
\begin{enumerate}
\item We have $H / S_T^{\delta_*} \in \call^1(\bbp)$ if and only if $H / N_T \in \call^1(\bbq)$.

\item If the equivalent conditions from (a) are fulfilled, then we have
\begin{align*}
\pi_t^{S^{\delta_*}}(H) = N_t \, \bbe_{\bbq} \bigg[ \frac{H}{N_T} \, \Big| \, \calf_t \bigg], \quad t \in [0,T].
\end{align*}
\end{enumerate}
\end{enumerate}
\end{proposition}

\begin{proof}
Let $S^{\delta}$ be a nonnegative self-financing portfolio. Since $(N,\bbq)$ is a num\'{e}raire pair, the process 
\begin{align*}
\frac{S^{\delta}}{N} = \frac{S^{\delta} Z}{N} \frac{1}{Z}
\end{align*}
is a $\bbq$-local martingale. Moreover, the process $1 / Z$ is the density process of $\bbp$ relative to $\bbq$. Therefore, by \cite[Prop.~III.3.8.b]{Jacod-Shiryaev} the benchmarked portfolio 
\begin{align*}
\hat{S}^{\delta} = \frac{S^{\delta}}{S^{\delta_*}} = \frac{N_0}{S_0^{\delta_*}} \frac{S^{\delta} Z}{N}
\end{align*}
is a $\bbp$-local martingale, proving the first statement. For the proof of the second statement, let $H$ be a nonnegative $\calf_T$-measurable random variable. Then, according to formula (III.3.9) on page 168 in \cite{Jacod-Shiryaev}, for all $t \in [0,T]$ we obtain
\begin{align*}
\pi_t^{S^{\delta_*}}(H) = S_t^{\delta_*} \, \bbe_{\bbp} \bigg[ \frac{H}{S_T^{\delta_*}} \, \Big| \, \calf_t \bigg] = \frac{N_t}{Z_t} \, \bbe_{\bbp} \bigg[ \frac{H}{N_T} Z_T \, \Big| \, \calf_t \bigg] = N_t \, \bbe_{\bbq} \bigg[ \frac{H}{N_T} \, \Big| \, \calf_t \bigg],
\end{align*}
which completes the proof.
\end{proof}

Now, we are interested in a converse statement of Proposition~\ref{prop-num-1}. In view of (\ref{GOP-num}), the natural candidate for the density process $Z$ is
\begin{align}\label{Z-candidate}
Z := \frac{S_0^{\delta_*}}{N_0}  \frac{N}{S^{\delta_*}}
\end{align}
for a given growth optimal portfolio $S^{\delta_*}$ and a strictly positive portfolio $N$.

\begin{proposition}\label{prop-num-2}
Let $S^{\delta_*}$ be a growth optimal portfolio, let $N$ be a strictly positive portfolio, and let $Z$ be the strictly positive supermartingale given by (\ref{Z-candidate}). Then the following statements are equivalent:
\begin{enumerate}
\item[(i)] There exists an equivalent probability measure $\bbq \sim \bbp$ on $(\Omega,\calf)$ such that $(N,\bbq)$ is a num\'{e}raire pair and the density process of $\bbq$ relative to $\bbp$ is given by $Z$.

\item[(ii)] $Z$ is a uniformly integrable martingale with $\bbp(Z_{\infty} > 0) = 1$.
\end{enumerate}
\end{proposition}

\begin{proof}
This is a consequence of \cite[Prop.~III.3.5]{Jacod-Shiryaev}.
\end{proof}

\begin{remark}
Note that Proposition~\ref{prop-num-2} has the following consequences.
\begin{itemize}
\item In general, an appropriate valuation measure $\bbq \sim \bbp$ does not exist. This shows that the converse of Proposition~\ref{prop-num-1} generally fails and illustrates the modeling freedom gained by the Benchmark Approach in comparison to the classical risk-neutral approach. Under the Benchmark Approach, we can choose models where the candidate (\ref{Z-candidate}) for the density process is a strict supermartingale rather than a uniformly integrable martingale. As argued in \cite{Platen}, such models reflect more realistically the long-term market evolution. This provides considerable freedom for an interest rate term structure model to admit an affine realization under the real-world probability measure.

\item If an appropriate valuation measure $\bbq \sim \bbp$ exists, then it is unique due to the specified form (\ref{Z-candidate}) of the density process. In this sense, we obtain a unique local martingale measure, yielding minimal possible prices, even if the market is incomplete.

\item Other pricing rules than real-world pricing are allowed under the Benchmark Approach, when they are yielding self-financing nonnegative portfolios that are replicating the given payoff. When benchmarked, these portfolios are local martingales and, thus, supermartingales, and they are more expensive than the real-world prices. These portfolios do not represent arbitrage portfolios under the Benchmark Approach.
\end{itemize}
\end{remark}

\section{The HJM model under the Benchmark Approach}\label{sec-HJM-benchmark}

In this section, we provide a review of bond markets arising from HJM interest rate term structure models under the Benchmark Approach. As in Section \ref{sec-bond-benchmark}, the upcoming definitions and results are well-known, and we provide them in order to keep our presentation self-contained and to introduce further notation, which we will need later. Our main references for this section are \cite{Christensen-Platen} and \cite{B-N-Platen} for HJM models under the Benchmark Approach, and \cite{Eberlein_J} for risk-neutral HJM models.

We fix nonnegative integers $d,n \in \bbn_0$ with $d+n \in \bbn$ and $m \in \{ 0,\ldots,n \}$. Let $W$ be a $\bbr^d$-valued standard Wiener processes, and let $X$ be a $\bbr^n$-valued pure jump L\'{e}vy process such that the canonical representations (see \cite[Thm.~II.2.34]{Jacod-Shiryaev}) of its components are given by $X^k = x * \mu^{X^k}$ for $k=1,\ldots,m$ and $X^k = x * (\mu^{X^k} - \nu^k)$ for $k=m+1,\ldots,n$, where $\nu^k$ denotes the compensator of $\mu^{X^k}$. We denote by $F$ the L\'{e}vy measure of $X$, and by $F^k$ the L\'{e}vy measure of $X^k$ for $k = 1,\ldots,n$. Denoting by $\nu$ the compensator of $\mu^X$, we have $\nu(dt,dx) = dt \otimes F(dx)$ and $\nu^k(dt,dx) = dt \otimes F^k(dx)$ for $k=1,\ldots,n$. We assume that the L\'{e}vy processes $X^1,\ldots,X^n$ are independent.

We fix an initial forward curve $f_0^* : \bbr_+ \to \bbr$, and volatilities $\sigma : \Omega \times \Delta \to \bbr^d$ and $\gamma : \Omega \times \Delta \to \bbr^n$, where $\Delta \subset \bbr^2$ denotes the set
\begin{align*}
\Delta := \{ (t,T) \in \bbr_+^2 : t \leq T \}.
\end{align*}
For the rest of this section, we impose the following conditions, which are typical for HJM type models.

\begin{assumption}\label{ass-HJM}
We suppose that the following conditions are satisfied:
\begin{enumerate}
\item $f_0^*$ is measurable and locally integrable.

\item The volatility $\sigma \mathbbm{1}_{\Delta}$ is $\calo \otimes \calb(\bbr_+)$-measurable and locally bounded\footnote[1]{Here, the term ``locally bounded'' means that the volatility is bounded on every bounded subset $\Gamma \subset \Delta$.}.

\item The volatility $\gamma \mathbbm{1}_{\Delta}$ is $\calp \otimes \calb(\bbr_+)$-measurable and locally bounded.
\end{enumerate}
\end{assumption}

Due to Assumption \ref{ass-HJM}, the integrated volatilities $\Sigma : \Omega \times \Delta \to \bbr^d$ and $\Gamma : \Omega \times \Delta \to \bbr^n$ defined as $\Sigma_t(T) := -\int_t^T \sigma_t(s) ds$ and $\Gamma_t(T) := -\int_t^T \gamma_t(s) ds$ are well-defined, the integrated volatility $\Sigma \mathbbm{1}_{\Delta}$ is $\calo \otimes \calb(\bbr_+)$-measurable, and the integrated volatility $\Gamma \mathbbm{1}_{\Delta}$ is $\calp \otimes \calb(\bbr_+)$-measurable. For what follows $\lambda$ denotes the Lebesgue measure on $\bbr_+$.

\begin{definition}\label{def-MPR-pair}
A pair $(\theta,\psi)$ is called a \emph{pair of market prices of risk}, if with $\phi := 1 - \psi$ the following conditions are satisfied:
\begin{enumerate}
\item $\theta : \Omega \times \bbr_+ \to \bbr^d$ is an optional process  such that $\| \theta \|_{\bbr^d}^2 \cdot \lambda \in \calv^+$.

\item $\psi : \Omega \times \bbr_+ \times \bbr \to (-\infty,1)^n$ is a predictable process such that for all $T \in \bbr_+$ we have
\begin{align}\label{int-Y-part-1}
&\big| x \phi^k(x) e^{x \Sigma^k(T)} \big| * \nu^k \in \calv^+, \quad k=1,\ldots,m,
\\ \label{int-Y-part-2} &\big| x \big( \phi^k(x) e^{x \Sigma^k(T)} - 1 \big) \big| * \nu^k \in \calv^+, \quad k = m+1,\ldots,n.
\end{align}

\item $\alpha^{(\theta,\phi)} \mathbbm{1}_{\Delta}$ is locally bounded, where $\alpha^{(\theta,\phi)} : \Omega \times \Delta \to \bbr$ is defined as
\begin{equation}\label{HJM-drift}
\begin{aligned}
\alpha_t^{(\theta,\phi)}(T) &= -\sum_{k=1}^d \sigma_t^k(T) ( \Sigma_t^k(T) - \theta_t^k )
\\ &\quad - \sum_{k=1}^m \gamma_t^k(T) \int_{\bbr} x \phi_t^k(x) e^{x \Gamma_t^k(T)} F^k(dx)
\\ &\quad - \sum_{k=m+1}^n \gamma_t^k(T) \int_{\bbr} x \big( \phi_t^k(x) e^{x \Gamma_t^k(T)} - 1 \big) F^k(dx).
\end{aligned}
\end{equation}
\end{enumerate}
\end{definition}

Now, we fix a pair of market prices of risk $(\theta,\psi)$ and define the $(0,\infty)^n$-valued process $\phi := 1 - \psi$. Furthermore, we define the drift $\alpha^{(\theta,\phi)}$ according to (\ref{HJM-drift}) and consider the HJM term structure model
\begin{align}\label{HJM-model}
f(T) = f_0^*(T) + \alpha^{(\theta,\phi)}(T) \cdot \lambda + \sigma(T) \cdot W + \gamma(T) \cdot X, \quad T \in \bbr_+.
\end{align}
For what follows, we suppose that for each $T \in \bbr_+$ the bond price process $P(T)$ given by
\begin{align}\label{bond-market}
P_t(T) = \exp \bigg( -\int_t^T f_t(s) ds \bigg), \quad t \in [0,T]
\end{align}
is a special semimartingale.

\begin{remark}
The first two points of Definition \ref{def-MPR-pair} ensure that $\alpha^{(\theta,\phi)}$ is well-defined, and that $\alpha^{(\theta,\phi)} \mathbbm{1}_{\Delta}$ is $\calo \otimes \calb(\bbr_+)$-measurable.
\end{remark}

\begin{remark}\label{remark-short-rate}
Note that Assumption \ref{ass-HJM} and Definition \ref{def-MPR-pair} imply that $f \mathbbm{1}_{\Delta}$ has a $\calo \otimes \calb(\bbr_+)$-measurable version, and that the short rate $R$ defined as $R_t := f(t,t)$, $t \in \bbr_+$ has an optional, locally integrable version.
\end{remark}

The following result shows that, subject to the additional integrability condition (\ref{cond-no-arb-portf}), the pair of market prices of risk $(\theta,\psi)$ gives rise to an arbitrage free bond market in the spirit of the Benchmark Approach.

\begin{proposition}\label{prop-GOP-HJM}
Suppose that
\begin{align}\label{cond-no-arb-portf}
\bigg[ \psi^2 + \bigg( \frac{\psi}{1 - \psi} \bigg)^2 \bigg] * \nu \in \calv^+.
\end{align}
Then the following statements are true:
\begin{enumerate}
\item There exists a growth optimal portfolio $S^{\delta_*}$.

\item No arbitrage portfolio exists.
\end{enumerate}
\end{proposition}

\begin{proof}[Sketch of the proof]
Defining the process $S^{\delta_*}$ as the stochastic exponential
\begin{align}\label{SDE-GOP-W-mu}
S^{\delta_*} = S_0^{\delta_*} \cale \bigg( \bigg( R + \| \theta \|_{\bbr^d}^2 + \Big \langle \frac{\psi}{1 - \psi}, \psi \Big\rangle_{L^2(F)} \bigg) \cdot \lambda + \theta \cdot W + \frac{\psi}{1 - \psi} * (\mu^X - \nu) \bigg),
\end{align}
we can verify that for every nonnegative self-financing portfolio $S^{\delta}$ the benchmarked portfolio $\hat{S}^{\delta} = S^{\delta} / S^{\delta_*}$ is a local martingale. Consequently, $S^{\delta_*}$ is a growth optimal portfolio, and by Proposition~\ref{prop-no-arbitrage} no arbitrage portfolio exists.
\end{proof}

\begin{remark}
The dynamics (\ref{SDE-GOP-W-mu}) of the growth optimal portfolio $S^{\delta_*}$ have been derived in~\cite{Christensen-Platen}.
\end{remark}

\begin{remark}\label{remark-MPR}
Performing the calculations indicated in the sketch of the proof of Proposition~\ref{prop-GOP-HJM}, we obtain that for each $T \in \bbr_+$ the bond price process $P(T)$ is of the form (\ref{bond-dyn-MPR}). As shown in \cite{Christensen-Platen}, the pair $(\theta,\psi)$ is a solution of the equation (\ref{MPR-solution}), which explains the terminology ``pair of market prices of risk''.
\end{remark}

Now, we review when the HJM interest rate model (\ref{HJM-model}) admits no arbitrage within the risk-neutral framework. Then the num\'{e}raire is the savings account $B := \exp (R \cdot \lambda)$,
and the bond market model is free of arbitrage if there exists a local martingale measure, that is, an equivalent probability measure $\bbq \sim \bbp$ such that for all $T \in \bbr_+$ the discounted bond prices $P(T)/B$ are $\bbq$-local martingales. 

\begin{proposition}\label{prop-NA}
Suppose that $\calf = \bigvee_{t \in \bbr_+} \calf_t$. Then the following statements are equivalent:
\begin{enumerate}
\item[(i)] There exists an equivalent local martingale measure $\bbq \sim \bbp$ on $(\Omega,\calf)$.

\item[(ii)] We have
\begin{align}\label{int-cond-NA-classical}
(|\psi| \wedge \psi^2) * \nu \in \calv^+, 
\end{align}
and the positive supermartingale $Z$ given by (\ref{Z-candidate-intro}) satisfies (\ref{uniformly}).
\end{enumerate}
If the previous conditions are satisfied, then $Z$ is the density process of $\bbq$ relative to $\bbp$.
\end{proposition}

\begin{proof}
For the proof we refer to \cite[Thm.~3.1]{Eberlein_J}. As mentioned there, it is a consequence of Girsanov's theorem (see \cite{Jacod-Shiryaev}), and it is also essentially contained in~\cite{BKR}.
\end{proof}

\begin{remark}
Now, we can compare the HJM model under real-world pricing of the Benchmark Approach and under the classical risk-neutral approach. In both approaches, we start with writing down the forward rate dynamics (\ref{HJM-model}) under the real-world probability measure $\bbp$; that is, we specify the volatilities $\sigma$ and $\gamma$, and the pair $(\theta,\psi)$ of market prices of risk. As pointed out in \cite{Becherer}, there are two dual approaches for finding num\'{e}raire pairs:
\begin{itemize}
\item Fix a process $N$ and find $\bbq \sim \bbp$ such that $(\bbq,N)$ is a num\'{e}raire pair. In the classical risk-neutral approach, this is done by choosing the savings account $B$ as candidate for the num\'{e}raire. Therefore, in case of existence, the num\'{e}raire pair is given by $(\bbq,B)$.

\item Fix an equivalent measure $\bbq \sim \bbp$ and find a process $N$ such that $(\bbq,N)$ is a num\'{e}raire pair. Under the Benchmark Approach, this is done with the real-world measure $\bbp$. In case of existence, the num\'{e}raire is a growth optimal portfolio $S^{\delta_*}$, and hence, the num\'{e}raire pair is given by $(\bbp,S^{\delta_*})$.
\end{itemize}
Apart from the slightly different integrability conditions (\ref{cond-no-arb-portf}) and (\ref{int-cond-NA-classical}), Propositions~\ref{prop-GOP-HJM} and \ref{prop-NA} show that the Benchmark Approach is more general in this respect. Under the Benchmark Approach, we can choose models where the candidate $Z$ for the density process given by (\ref{Z-candidate-intro}) is a strict supermartingale rather than a uniformly integrable martingale. As argued in \cite{Platen}, such models reflect more realistically the long-term market evolution. 
\end{remark}

Here is an example of a class of interest rate models where $Z$ is a strict supermartingale, even with terminal value $Z_{\infty} = 0$. For details concerning squared Bessel processes we refer to \cite[Sec. 8.7]{Platen} and references therein.

\begin{example}\label{example-squared-Bessel}
For simplicity, we consider a HJM term structure model of the form
\begin{align*}
f(T) = f_0^*(T) + \alpha^{\theta}(T) \cdot \lambda + \sigma(T) \cdot W, \quad T \in \bbr_+
\end{align*}
driven by a one-dimensional standard Wiener process $W$. Let $y_0 \in (0,\infty)$ be arbitrary and set $b_0 := 1 / y_0$. Then the SDE
\begin{align}
\left\{
\begin{array}{rcl}
dB_t & = & 4 dt + 2 \sqrt{B_t} dW_t \medskip
\\ B_0 & = & b_0
\end{array}
\right.
\end{align}
has a unique positive solution $B$, which is called a squared Bessel process of dimension four. Moreover, the process $Y := 1 / B$ is a solution to the SDE
\begin{align}\label{SDE-Z-Bessel}
\left\{
\begin{array}{rcl}
dY_t & = & -2 Y_t^{3/2} dW_t \medskip
\\ Y_0 & = & y_0
\end{array}
\right.
\end{align}
and it is a positive local martingale, which is a strict supermartingale, with terminal variable $Y_{\infty} = 0$. We define the market price of risk $\theta := 2 \sqrt{Y}$. Denoting by $Z := \cale(-\theta \cdot W)$ the candidate for the density process, we obtain $Z = Y$, because, denoting by $\call$ the stochastic logarithm, we have
\begin{align*}
Z &= \cale(-\theta \cdot W) = \cale(-2 \sqrt{Y} \cdot W) = \cale \bigg( \frac{1}{Y} \cdot \big( -2 Y^{3/2} \cdot W \big) \bigg) 
\\ &= \cale \bigg( \frac{1}{Y} \cdot Y \bigg) = \cale(\call(Y)) = Y.
\end{align*}
Therefore, according to Proposition~\ref{prop-NA}, no equivalent local martingale measure exists. However, by Proposition~\ref{prop-GOP-HJM} no arbitrage portfolio exists, and hence, the bond market based on real-world pricing is free of arbitrage in the spirit of the Benchmark Approach.
\end{example}

We can extend Example~\ref{example-squared-Bessel} by adding pure jump L\'{e}vy processes. The upcoming example provides such an extension with an additional driving Poisson process.

\begin{example}\label{example-squared-Bessel-2}
We consider a HJM term structure model of the form (\ref{HJM-model}) with a one-dimensional standard Wiener process $W$ and a one-dimensional standard Poisson process $X$. Denoting by $Y$ the solution to the SDE (\ref{SDE-Z-Bessel}), we define the pair of market prices of risk $(\theta,\psi)$ by $\theta := 2 \sqrt{Y}$ and $\psi(x) := -Y$ for $x \in \bbr$, and the candidate for the density process $Z := \cale(-\theta \cdot W - \psi * (\mu^X - \nu))$. Then we have
\begin{align*}
Z = Y \, \cale(- \psi * (\mu^X - \nu)),
\end{align*}
and hence $Z_{\infty} = 0$. As in Example~\ref{example-squared-Bessel}, no equivalent local martingale measure exists, but the bond market based on real-world pricing is free of arbitrage in the spirit of the Benchmark Approach.
\end{example}

Note that for the particular choice $(\theta,\phi) = (0,1)$, or equivalently $(\theta,\psi) = (0,0)$, the drift term (\ref{HJM-drift}) becomes the well-known HJM drift condition
\begin{equation}\label{HJM-drift-classical}
\begin{aligned}
\alpha_t^{(0,1)}(T) &= -\sum_{k=1}^d \sigma_t^k(T) \Sigma_t^k(T)
- \sum_{k=1}^m \gamma_t^k(T) \int_{\bbr} x e^{x \Gamma_t^k(T)} F^k(dx)
\\ &\quad - \sum_{k=m+1}^n \gamma_t^k(T) \int_{\bbr} x \big( e^{x \Gamma_t^k(T)} - 1 \big) F^k(dx).
\end{aligned}
\end{equation}
Let us review the dynamics of the HJM interest rate term structure model (\ref{HJM-model}) after an equivalent measure change.

\begin{proposition}\label{prop-NA-dynamics}
Suppose that the conditions of Proposition~\ref{prop-NA} are fulfilled. Then, under the local martingale measure $\bbq$, we have
\begin{align*}
f(T) = f_0^*(T) + \bar{\alpha}(T) \cdot \lambda + \sigma(T) \cdot \bar{W} + \gamma(T) \cdot \bar{X}, \quad T \in \bbr_+,
\end{align*}
where $\bar{W}$ is a $\bbr^d$-valued standard Wiener processes, the process $\bar{X}$ is a $\bbr^n$-valued pure jump semimartingale with compensator $\phi_t(x) F(dx) dt$, and the drift is given by
\begin{equation}\label{alpha-Q}
\begin{aligned}
\bar{\alpha}_t(T) &= -\sum_{k=1}^d \sigma_t^k(T) \Sigma_t^k(T) - \sum_{k=1}^m \gamma_t^k(T) \int_{\bbr} x Y_t^k(x) e^{x \Gamma_t^k(T)} F^k(dx)
\\ &\quad - \sum_{k=m+1}^n \gamma_t^k(T) \int_{\bbr} x Y_t^k(x) \big( e^{x \Gamma_t^k(T)} - 1 \big) F^k(dx).
\end{aligned}
\end{equation}
\end{proposition}

\begin{proof}
This follows from \cite[Thm.~3.1]{Eberlein_J} and its proof.
\end{proof}

\begin{remark}\label{remark-measure-change}
Let us distinguish two cases arising in Proposition~\ref{prop-NA-dynamics}:
\begin{itemize}
\item If the HJM model (\ref{HJM-model}) is only driven by Wiener processes, then we have $\bar{\alpha} = \alpha^{0}$; that is, the drift term after the measure change coincides with the classical HJM drift term. Consequently, for every choice of the market price of risk $\theta$, the dynamics of the model can, in a broad sense, be regarded as that after an equivalent measure change; of course, only up to the technical requirements from Proposition~\ref{prop-NA}, in particular condition (\ref{uniformly}) concerning the density process.

\item If the HJM model (\ref{HJM-model}) has driving jump terms, then, in general, we have $\bar{\alpha} \neq \alpha^{(0,1)}$; that is, the drift term after the measure change differs from the classical HJM drift term. The reason for this observation is that, in contrast to the Wiener processes, the L\'{e}vy processes with jumps have different characteristics after the measure change. Consequently, changing the market price of risk can no longer be interpreted as an equivalent measure change; not even in the broad sense where we disregard condition (\ref{uniformly}).
\end{itemize}
These considerations will provide geometric interpretations of our upcoming results concerning the existence of affine realizations; see Remarks~\ref{remark-geometric} and \ref{remark-geometric-2} below.
\end{remark}

\section{Affine realizations for SPDEs driven by L\'{e}vy processes}\label{sec-affine-real-SPDE}

In this section, we provide results on invariant foliations for
SPDEs driven by L\'evy processes, which we will apply to the HJMM equation (\ref{HJMM}) later on. We refer to \cite[Sec. 2 and 3]{Tappe-Wiener} and \cite[Sec. 2]{Tappe-Levy} for more details and explanations about invariant foliations.

Fix a positive integer $n \in \bbn$ and let $X$ be a $\bbr^n$-valued L\'{e}vy process with independent components. In order to avoid trivialities, we assume that $c^k + F^k(\mathbb{R}) > 0$ for $k=1,\ldots,n$, where $c^k \in \bbr_+$ denotes the Gaussian part, and $F^k$ the L\'evy measure. 

Let $\mathcal{Y}$ be a nonempty topological space and let $(Y^{y})_{y \in \mathcal{Y}}$ be a family of $\mathcal{Y}$-valued, adapted and c\`{a}dl\`{a}g processes with $Y_0^{y} = y$ for all $y \in \mathcal{Y}$. We shall deal with SPDEs of the type
\begin{align}\label{SPDE-manifold}
\left\{
\begin{array}{rcl}
dr_t & = & \big( A r_t + \alpha(r_t,Y_t) \big) dt + \sigma(r_{t-})dX_t
\medskip
\\ r_0 & = & h_0 \medskip
\\ Y_0 & = & y_0
\end{array}
\right.
\end{align}
on a separable Hilbert space $H$. In (\ref{SPDE-manifold}), the operator $A : \mathcal{D}(A) \subset H \rightarrow H$ is the infinitesimal
generator of a $C_0$-semigroup $(S_t)_{t \geq 0}$ on $H$, and $\alpha : H \times \mathcal{Y} \rightarrow H$ and $\sigma : H \rightarrow H^n$ are measurable mappings. In the spirit of \cite{P-Z-book}, for $h_0 \in H$ and $y_0 \in \mathcal{Y}$ we call an $H$-valued c\`{a}dl\`{a}g adapted process $(r_t)_{t \geq 0}$ a \textit{weak solution} to (\ref{SPDE-manifold}) with $r_0 = h_0$ and $Y_0 = y_0$, if for each $\zeta \in \mathcal{D}(A^*)$ we have
\begin{align*}
\langle \zeta,r_t \rangle_H &= \langle \zeta,h_0 \rangle_H + \int_0^t \big( \langle A^* \zeta, r_s \rangle_H + \langle \zeta, \alpha(r_s,Y_s^{y_0}) \rangle_H \big) ds 
\\ &\quad + \sum_{k=1}^n \int_0^t \langle \zeta,\sigma^k(r_{s-}) \rangle_H dX_s^k, \quad t \geq 0,
\end{align*}
where $\langle \cdot,\cdot \rangle_H$ denotes the inner product of the Hilbert space $H$. 

\begin{remark}
In the context of the HJMM equation (\ref{HJMM}), the family $(Y^{y})_{y \in \mathcal{Y}}$ represents the source providing the market price of risk processes. More precisely, in the following Section \ref{sec-affine-real-general} we will fix deterministic mappings $\Theta : \caly \to \bbr^d$ and $\Psi : \caly \times \bbr \to (-\infty,1)^n$, and we define the market price of risk as $(\theta,\psi) := (\Theta(Y^{y_0}),\Psi(Y^{y_0}))$ for any starting point $y_0 \in \caly$. As there are no restrictions on the family $(Y^{y})_{y \in \mathcal{Y}}$, this provides a general class of market price of risk processes, and the parametric form will be convenient for technical purposes, for example in Section \ref{sec-nec-Psi}, when we will prove the announced result that under certain assumptions the market price of risk must be constant.
\end{remark}

Throughout this section, we impose the following regularity conditions, which ensure existence and uniqueness of weak solutions to (\ref{SPDE-manifold}).

\begin{assumption}
We suppose that there exist constants $K,L > 0$ such that
\begin{align}\label{alpha-lin-gr}
\| \alpha(h,y) \|_H &\leq K (1 + \| h \|_H)
\end{align}
for all $h \in H$ and $y \in \mathcal{Y}$, and
\begin{align}\label{alpha-Lipschitz} 
\| \alpha(h_1,y) - \alpha(h_2,y) \|_H &\leq L \| h_1 - h_2 \|_H,
\\ \label{sigma-Lipschitz} \| \sigma^k(h_1) - \sigma^k(h_2) \|_H &\leq L \| h_1 - h_2 \|_H, \quad k=1,\ldots,n
\end{align}
for all $h_1,h_2 \in H$ and $y \in \mathcal{Y}$.
\end{assumption}

In what follows, let $V \subset H$ be a finite dimensional linear
subspace.

\begin{definition}
A family $(\mathcal{M}_t)_{t \geq 0}$ of affine subspaces
$\mathcal{M}_t \subset H$, $t \geq 0$ is called a \emph{foliation
generated by $V$}, if there exists $\psi \in C^1(\mathbb{R}_+;H)$
such that
\begin{align*}
\mathcal{M}_t = \psi(t) + V, \quad t \geq 0.
\end{align*}
The map $\psi$ is called a \emph{parametrization} of the foliation $(\mathcal{M}_t)_{t \geq 0}$.
\end{definition}

In what follows, let $(\mathcal{M}_t)_{t \geq 0}$ be a foliation
generated by the subspace $V$.

\begin{definition}\label{def-inv-foliation-pre}
Let $y_0 \in \caly$ be arbitrary. The foliation $(\mathcal{M}_t)_{t \geq 0}$ is called \emph{invariant for (\ref{SPDE-manifold}) with $Y_0 = y_0$} if for all $t_0 \in
\mathbb{R}_+$ and $h_0 \in \mathcal{M}_{t_0}$ the weak solution $r$ to (\ref{SPDE-manifold}) with $r_0 = h_0$ and $Y_0 = y_0$ satisfies
\begin{align*}
\mathbb{P}(r_t \in \mathcal{M}_{t_0 + t}) = 1 \quad \text{for all $t \in \bbr_+$.}
\end{align*}
\end{definition}

\begin{definition}\label{def-inv-foliation}
The foliation $(\mathcal{M}_t)_{t \geq 0}$ is called \emph{invariant for (\ref{SPDE-manifold})}, if for every $y_0 \in \caly$ it is invariant for (\ref{SPDE-manifold}) with $Y_0 = y_0$.
\end{definition}

The proofs of the following two results are similar to the corresponding results in \cite[Sec. 2]{Tappe-Levy} (see also \cite[Sec. 2]{Tappe-Wiener}), and are therefore omitted. In the following, the subspace $T \mathcal{M}_t := \psi'(t) + V$ denotes the tangent space of the foliation at time $t$. Its definition does not depend on the choice of the parametrization $\psi$; see \cite{Tappe-Wiener}.

\begin{theorem}\label{thm-foliation-1-pre}
Let $y_0 \in \caly$ be arbitrary, and suppose that the foliation $(\mathcal{M}_t)_{t \geq 0}$ is invariant for (\ref{SPDE-manifold}) with $Y_0 = y_0$. Then we have
\begin{align}\label{domain-pre-pre}
\mathcal{M}_t &\subset \mathcal{D}(A), \quad t \in \bbr_+,
\\ \label{nu-pre-pre} A h + \alpha ( h,y_0 ) &\in T \mathcal{M}_t, \quad t \in \bbr_+ \text{ and } h \in \mathcal{M}_t,
\\ \label{sigma-pre-pre} \sigma^k ( h ) &\in V, \quad t \in \bbr_+, h \in
\mathcal{M}_t \text{ and } k=1,\ldots,n.
\end{align}
\end{theorem}

\begin{theorem}\label{thm-foliation-1}
The following statements are equivalent:
\begin{enumerate}
\item[(i)] The foliation $(\mathcal{M}_t)_{t \geq 0}$ is invariant
for (\ref{SPDE-manifold}).

\item[(ii)] We have
\begin{align}\label{domain-pre}
\mathcal{M}_t &\subset \mathcal{D}(A), \quad t \in \bbr_+
\\ \label{nu-pre} A h + \alpha ( h,y ) &\in T \mathcal{M}_t, \quad t \in \bbr_+ \text{ and } (h,y) \in \mathcal{M}_t \times \mathcal{Y},
\\ \label{sigma-pre} \sigma^k ( h ) &\in V, \quad t \in \bbr_+, h \in
\mathcal{M}_t \text{ and } k=1,\ldots,n.
\end{align}
\end{enumerate}
\end{theorem}

\begin{definition}
Let $y_0 \in \caly$ be arbitrary. 
\begin{enumerate}
\item Let $V \subset H$ be a finite dimensional subspace. The SPDE (\ref{SPDE-manifold}) with $Y_0 = y_0$ has an \emph{affine realization generated by $V$} if for each $h_0 \in \mathcal{D}(A)$ there exists a foliation $(\mathcal{M}_t)_{t \geq 0}$ generated by $V$ with $h_0 \in
\mathcal{M}_0$, which is invariant for (\ref{SPDE-manifold}) with $Y_0 = y_0$.

\item The SPDE (\ref{SPDE-manifold}) with $Y_0 = y_0$ has an \emph{affine realization} if it has an affine realization with $Y_0 = y_0$ generated by some finite dimensional subspace $V$.
\end{enumerate}
\end{definition}

\begin{definition}\mbox{}
\begin{enumerate}
\item Let $V \subset H$ be a finite dimensional subspace. The SPDE (\ref{SPDE-manifold}) has an \emph{affine realization generated by
$V$} if for each $h_0 \in \mathcal{D}(A)$ there exists a foliation
$(\mathcal{M}_t)_{t \geq 0}$ generated by $V$ with $h_0 \in
\mathcal{M}_0$, which is invariant for (\ref{SPDE-manifold}). Note that, according to Definition~\ref{def-inv-foliation}, the latter condition means that for every $y_0 \in \caly$ the foliation $(\mathcal{M}_t)_{t \geq 0}$ it is invariant for (\ref{SPDE-manifold}) with $Y_0 = y_0$.

\item The SPDE (\ref{SPDE-manifold}) has an \emph{affine realization} if it has an affine realization generated by some finite dimensional subspace $V$.
\end{enumerate}
\end{definition}

From now on, we fix an element $y^* \in \caly$, and deal with the question when the existence of an affine realization with $Y_0 = y^*$ implies the existence of an affine realization. For this purpose, we define the subspace $U_{y^*} \subset H$ as
\begin{align}\label{def-U-star}
U_{y^*} := \langle \alpha(h,y) - \alpha(h,y^*) : h \in H \text{ and } y \in \caly \rangle.
\end{align}
Here, and in the sequel, we denote by $\langle B \rangle$ the linear space generated by some subset $B \subset H$. There is no danger of confusion with the inner product of the Hilbert space, which we denote by $\langle \cdot,\cdot \rangle_H$.

\begin{theorem}\label{thm-affine-real}
The following statements are equivalent:
\begin{enumerate}
\item[(i)] The SPDE (\ref{SPDE-manifold}) has an affine realization.

\item[(ii)] The SPDE (\ref{SPDE-manifold}) with $Y_0 = y^*$ has an affine realization, and $U_{y^*}$ is a finite dimensional subspace of $\cald(A)$.
\end{enumerate}
\end{theorem}

\begin{proof}
(i) $\Rightarrow$ (ii): By hypothesis, the SPDE (\ref{SPDE-manifold}) with $Y_0 = y^*$ has an affine realization. Let $V$ be a finite dimensional subspace generating the affine realization, and let $h_0 \in \cald(A)$ be arbitrary. Then there exists a foliation $(\calm_t)_{t \geq 0}$ generated by $V$ with $h_0 \in \calm_0$, which is invariant for (\ref{SPDE-manifold}). By condition (\ref{nu-pre}) of Theorem~\ref{thm-foliation-1} we obtain
\begin{align*}
A h + \alpha(h,y) \in T \calm_0 \quad \text{for all $h \in \calm_0$ and all $y \in \caly$.}
\end{align*}
In particular, for all $h \in \calm_0$ and all $y \in \caly$ we get
\begin{align*}
\alpha(h,y) - \alpha(h,y^*) = \big( A h + \alpha(h,y) \big) - \big( A h + \alpha(h,y^*) \big) \in V.
\end{align*}
Therefore, we deduce
\begin{align*}
\alpha(h,y) - \alpha(h,y^*) \in V \quad \text{for all $h \in \cald(A)$ and all $y \in \caly$.}
\end{align*}
Since $\alpha(\cdot,y)$ is continuous for each $y \in \caly$, the domain $\cald(A)$ is dense in $H$, and $V$ is closed, we conclude that
\begin{align*}
\alpha(h,y) - \alpha(h,y^*) \in V \quad \text{for all $h \in H$ and all $y \in \caly$,}
\end{align*}
which proves that the subspace $U_{y^*}$ is finite dimensional. Furthermore, by condition (\ref{domain-pre}) of Theorem~\ref{thm-foliation-1}, it is contained in $\cald(A)$. 

\noindent(ii) $\Rightarrow$ (i): There exists a finite dimensional subspace $V$ generating an affine realization for (\ref{SPDE-manifold}) with $Y_0 = y^*$. Let $h_0 \in \cald(A)$ be arbitrary. Then there exists a foliation $(\calm_t)_{t \geq 0}$ generated by $V$ with $h_0 \in \calm_0$, which is invariant for (\ref{SPDE-manifold}) with $Y_0 = y^*$. According to Theorem~\ref{thm-foliation-1-pre}, for all $t \in \bbr_+$ we have
\begin{align}\label{inv-star-1}
\calm_t &\subset \cald(A),
\\ \label{inv-star-2} A h + \alpha(h,y^*) &\in T \calm_t, \quad h \in \calm_t,
\\ \label{inv-star-3} \sigma^k(h) &\in V, \quad h \in \calm_t \text{ and } k=1,\ldots,n.
\end{align}
Since $U_{y^*}$ is a finite dimensional subspace of $\cald(A)$, by the just derived relation (\ref{inv-star-1}) the subspace
\begin{align*}
\bar{V} := V + U_{y^*}
\end{align*}
is a finite dimensional subspace of $\cald(A)$, too. We define the new foliation $(\bar{\calm}_t)_{t \geq 0}$ as $\bar{\calm}_t := \calm_t + \bar{V}$. Then, by (\ref{inv-star-1}) and (\ref{inv-star-3}), for all $t \in \bbr_+$ we obtain
\begin{align*}
\bar{\calm}_t &\subset \cald(A),
\\ \sigma^k(h) &\in \bar{V}, \quad h \in \calm_t \text{ and } k=1,\ldots,n.
\end{align*}
Moreover, by (\ref{inv-star-2}), for all $t \in \bbr_+$ and $(h,y) \in \calm_t \times \caly$ we obtain
\begin{align*}
A h + \alpha(h,y) = \big( A h + \alpha(h,y^*) \big) + \big( \alpha(h,y) - \alpha(h,y^*) \big) \in T \calm_t + \bar{V} = T \bar{\calm}_t. 
\end{align*}
Consequently, by Theorem~\ref{thm-foliation-1} the SPDE (\ref{SPDE-manifold}) has an affine realization generated by $\bar{V}$.
\end{proof}

\begin{proposition}\label{prop-sigma-nec}
Suppose that the SPDE (\ref{SPDE-manifold}) with $Y_0 = y^*$ has an affine realization. Then the subspace $U \subset H$ defined as $U := \sum_{k=1}^n \langle \sigma^k(H) \rangle$ is a finite dimensional subspace of $\cald(A)$.
\end{proposition}

\begin{proof}
There exists a finite dimensional subspace $V$ generating the affine realization. Let $h_0 \in \cald(A)$ be arbitrary. Then there exists a foliation $(\calm_t)_{t \geq 0}$ generated by $V$ with $h_0 \in \calm_0$, which is invariant for (\ref{SPDE-manifold}). By condition (\ref{sigma-pre-pre}) of Theorem~\ref{thm-foliation-1-pre} we obtain
\begin{align*}
\sigma^k(h_0) \in V \quad \text{for all $k=1,\ldots,n$.}
\end{align*}
Since the volatilities $\sigma^k$, $k=1,\ldots,n$, are continuous, the domain $\cald(A)$ is dense in $H$, and $V$ is closed, we conclude that
\begin{align*}
\sigma^k(H) \subset V \quad \text{for all $k=1,\ldots,n$,}
\end{align*}
which proves that the subspace $U$ is finite dimensional. Moreover, by condition (\ref{domain-pre-pre}) of Theorem~\ref{thm-foliation-1-pre} we have $V \subset \cald(A)$, which proves that $U$ is contained in $\cald(A)$.
\end{proof}

\section{Affine realizations for the HJMM equation with real-world forward rate dynamics}\label{sec-affine-real-general}

In this section, we start our analysis regarding the existence of affine realizations for the HJMM equation with real-world forward rate dynamics, and present our first main result, which establishes the connection between the existence of affine realizations for the HJMM equation with real-world forward rate dynamics and for the classical HJMM equation based on risk-neutral pricing.

First, we introduce the space of forward curves.
We fix a nondecreasing $C^1$-function $w : \bbr_+ \to [1,\infty)$ such that $w^{-1/3} \in \call^1(\bbr_+)$, and denote by $H$ the space of all absolutely continuous functions $h : \mathbb{R}_+
\rightarrow \mathbb{R}$ such that
\begin{align*}
\| h \|_{H} := \bigg( |h(0)|^2 + \int_{\mathbb{R}_+} |h'(\xi)|^2
w(\xi) d\xi \bigg)^{1/2} < \infty.
\end{align*}
Spaces of this kind have been utilized in \cite{fillnm}, to which we refer for their properties. Denoting by $(S_t)_{t \geq 0}$ the translation semigroup on $H$, the HJMM equation (\ref{HJMM}) is a particular example of the SPDE (\ref{SPDE-manifold}) with infinitesimal generator $A = d/d\xi$ on the domain
\begin{align*}
\cald(d/d\xi) = \{ h \in H \cap C^1(\bbr_+) : h' \in H \}.
\end{align*}
Next, we present our standing assumptions which prevail throughout this paper. As in Section \ref{sec-HJM-benchmark}, we fix Wiener processes $W^1,\ldots,W^d$ and pure jump L\'{e}vy processes $X^1,\ldots,X^n$, and, as in Section \ref{sec-affine-real-SPDE}, let $\mathcal{Y}$ be a nonempty topological space and let $(Y^{y})_{y \in \mathcal{Y}}$ be a family of $\mathcal{Y}$-valued, adapted and c\`{a}dl\`{a}g processes with $Y_0^{y} = y$ for all $y \in \mathcal{Y}$. 
Let $\sigma : H \to H^d$, $\gamma : H \to H^n$ and $\Theta : \caly \to \bbr^d$, $\Psi : \caly \times \bbr \to (-\infty,1)^n$ be measurable mappings. We define $\Phi : \caly \times \bbr \to (0,\infty)^n$ as $\Phi := 1 - \Psi$.

\begin{assumption}\label{ass-Lipschitz-HJMM}
We suppose that the following conditions are satisfied:
\begin{enumerate}
\item $\alpha : H \times \caly \to H$ given by (\ref{alpha-HJM}) satisfies the linear growth condition (\ref{alpha-lin-gr}) and the Lipschitz condition (\ref{alpha-Lipschitz}).

\item $\sigma^1,\ldots,\sigma^d$ and $\gamma^1,\ldots,\gamma^n$ are Lipschitz continuous.

\item For each $y_0 \in \caly$ the pair $(\theta,\psi) = (\Theta(Y^{y_0}),\Psi(Y^{y_0}))$ is a pair of market prices of risk satisfying the integrability condition (\ref{cond-no-arb-portf}). 

\item There are $p,q \in [1,\infty]$ with $\frac{1}{p} + \frac{1}{q} = 1$ such that $\Psi^k(y,\cdot) \in \call^p(F^k)$ and $x \mapsto e^{x \Gamma^k(h)}$ belongs to $\call^q(F^k;H)$ for all $k=1,\ldots,n$.

\item For all $(h,y) \in H \times \caly$ and $k = 1,\ldots,n$ the mapping
\begin{align*}
\int_{\bbr} \Psi^k(y,x) e^{x \Gamma^k(h)} F^k(dx)
\end{align*}
belongs to $\cald((d/d\xi)^2)$ with derivative
\begin{align*}
\frac{d}{d \xi} \int_{\bbr} \Psi^k(y,x) e^{x \Gamma^k(h)} F^k(dx) = -\gamma^k(h) \int_{\bbr} x \Psi^k(y,x) e^{x \Gamma^k(h)} F^k(dx).
\end{align*}
\end{enumerate}
\end{assumption}

After extending the state space $\caly$, if necessary, we may assume that there exists an element $y^* \in \caly$ such that $(\Theta(y^*),\Psi(y^*,\cdot)) = 0$ and $Y^{y^*} = y^*$.

\begin{remark}
Note that for all $h \in H$ we have
\begin{equation}\label{HJM-clasical}
\begin{aligned}
\alpha(h,y^*) &= - \sum_{k=1}^d \sigma^k(h) \Sigma^k(h) - \sum_{k=1}^m \gamma^k(h) \int_{\bbr} x e^{x \Gamma^k(h)} F^k(dx)
\\ &\quad - \sum_{k=m+1}^n \gamma^k(h) \int_{\bbr} x \big( e^{x \Gamma^k(h)} - 1 \big) F^k(dx).
\end{aligned}
\end{equation}
Within the Benchmark Approach, we always work with the num\'{e}raire pair $(S^{\delta_*},\bbp)$, but point out that (\ref{HJM-clasical}) is just the classical HJM drift term, which also occurs for num\'{e}raire pairs $(B,\bbq)$, where $B$ is the savings account and $\bbq \sim \bbp$ is a risk-neutral measure. In this sense, the existence of an affine realization with $Y_0 = y^*$ corresponds to the existence of an affine realization for classical HJM interest rate models, and for this situation, many results have been established in the literature.
\end{remark}

It is clear that the existence of an affine realization implies the existence of an affine realization with $Y_0 = y^*$. In order to provide a closer connection between these two types of realizations, we introduce the subspace $U_{\Psi,\gamma} \subset H$ by (\ref{def-U-Psi-gamma}).

\begin{theorem}\label{thm-HJMM-affine-real}
The following statements are equivalent:
\begin{enumerate}
\item[(i)] The HJMM equation (\ref{HJMM}) has an affine realization.

\item[(ii)] The HJMM equation (\ref{HJMM}) with $Y_0 = y^*$ has an affine realization, and we have $\dim U_{\Psi,\gamma} < \infty$.
\end{enumerate}
\end{theorem}

\begin{proof}
Suppose that the HJMM equation (\ref{HJMM}) with $Y_0 = y^*$ has an affine realization. Then, by Proposition~\ref{prop-sigma-nec} the subspace $\sum_{k=1}^d \langle \sigma^k(h) \rangle$ is a finite dimensional subspace of $\cald(d/d\xi)$. Moreover, by (\ref{alpha-HJM}) and (\ref{HJM-clasical}), for all $(h,y) \in H \times \caly$ we have
\begin{align*}
\alpha(h,y) - \alpha(h,y^*) &= \sum_{k=1}^d \Theta^k(y) \sigma^k(h) - \sum_{k=1}^n \gamma^k(h) \int_{\bbr} x ( \Phi^k(y,x) - 1 ) e^{\Gamma^k(h)} F^k(dx)
\\ &= \sum_{k=1}^d \Theta^k(y) \sigma^k(h) + \sum_{k=1}^n \gamma^k(h) \int_{\bbr} x \Psi^k(y,x) e^{\Gamma^k(h)} F^k(dx).
\end{align*}
Therefore, and since $\sum_{k=1}^d \langle \sigma^k(h) \rangle$ is a finite dimensional subspace of $\cald(d/d\xi)$, the subspace $U_{y^*}$ defined in (\ref{def-U-star}) is a finite dimensional subspace of $\cald(d/d\xi)$ if and only if the subspace
\begin{align*}
\Big\langle \sum_{k=1}^n \gamma^k(h) \int_{\bbr} x \Psi^k(y,x) e^{\Gamma^k(h)} F^k(dx) : h \in H \text{ and } y \in \caly \Big\rangle
\end{align*}
is a finite dimensional subspace of $\cald(d/d\xi)$. By virtue of Assumption \ref{ass-Lipschitz-HJMM}, this is fulfilled if and only if $U_{\Psi,\gamma}$ is a finite dimensional subspace of $\cald((d/d\xi)^2)$. 
\end{proof}

Theorem~\ref{thm-HJMM-affine-real} demonstrates the difference between the existence of affine realizations for classical HJM models and for HJM models with real-world forward rate dynamics. The crucial point is the subspace $U_{\Psi,\gamma}$, which has to be finite dimensional. Note that this condition only concerns the volatility $\gamma$ and the market price of risk $\Psi$ of the discontinuous part, but neither the volatility $\sigma$ nor the market price of risk $\Theta$ of the continuous part.

In view of Theorem~\ref{thm-HJMM-affine-real}, it will be useful to provide a result which gives sufficient conditions for the existence of an affine realization for the HJMM equation (\ref{HJMM}) with $Y_0 = y^*$. For this purpose, we recall that a mapping $\sigma : H \to \mathcal{D}((d/d\xi)^{\infty})$ is called \emph{quasi-exponential}, if
\begin{align*}
\dim \langle (d/d\xi)^m \sigma(h) : h \in H \text{ and } m \in \bbn_0 \rangle < \infty.
\end{align*}

\begin{proposition}\label{prop-suff-risk-neutral}
Suppose that the following conditions are satisfied:
\begin{enumerate}
\item $\sigma^1,\ldots,\sigma^d$ are quasi-exponential.

\item $\gamma^1,\ldots,\gamma^n$ are constant and quasi-exponential.
\end{enumerate}
Then the HJMM equation (\ref{HJMM}) with $Y_0 = y^*$ has an affine realization.
\end{proposition}

\begin{proof}
Since $\sigma^1,\ldots,\sigma^d$ and $\gamma^1,\ldots,\gamma^n$ are quasi-exponential, the subspace
\begin{align*}
V &:= \sum_{i=1}^d \langle (d/d\xi)^m \sigma^i(h) : h \in H \text{ and } m \in \bbn_0 \rangle 
\\ &\quad + \sum_{j=1}^n \langle (d/d\xi)^m \gamma^j(h) : h \in H \text{ and } m \in \bbn_0 \rangle
\end{align*}
is finite dimensional. Therefore, combining the arguments from \cite[Prop.~6.2]{Tappe-Wiener} and \cite[Thm.~5.1]{Tappe-Levy} shows that the HJMM equation (\ref{HJMM}) with $Y_0 = y^*$ has an affine realization generated by $V$.
\end{proof}

\section{Affine realizations for the HJMM equation driven by Wiener processes}\label{sec-Wiener}

In this section, we study the existence of affine realizations for the HJMM equation with real-world forward rate dynamics driven by Wiener processes. Then the corresponding HJMM equation (\ref{HJMM}) is of the particular form
\begin{align}\label{HJMM-Wiener}
\left\{
\begin{array}{rcl}
dr_t & = & \big( \frac{d}{d\xi} r_t + \alpha(r_t,Y_t) \big) dt + \sigma(r_{t})dW_t
\medskip
\\ r_0 & = & h_0 \medskip
\\ Y_0 & = & y_0
\end{array}
\right.
\end{align}
with a $\bbr^d$-valued standard Wiener process $W$.

\begin{theorem}\label{thm-HJMM-Wiener}
The following statements are equivalent:
\begin{enumerate}
\item[(i)] The HJMM equation (\ref{HJMM-Wiener}) has an affine realization.

\item[(ii)] The HJMM equation (\ref{HJMM-Wiener}) with $Y_0 = y^*$ has an affine realization.
\end{enumerate}
\end{theorem}

\begin{proof}
This is a direct consequence of Theorem~\ref{thm-HJMM-affine-real}.
\end{proof}

Theorem~\ref{thm-HJMM-Wiener} allows us to draw the following important conclusion: For Wiener process driven term structure models, all known results concerning the existence of affine realizations for HJM interest rate models under an assumed risk-neutral probability measure, see e.g. \cite{Bj_Sv,Bj_La,Filipovic,Tappe-Wiener}, transfer to interest rate models with real-world forward rate dynamics. In particular, we have the following result.

\begin{corollary}\label{cor-HJMM-Wiener}
Suppose that $\sigma^1,\ldots,\sigma^d$ are quasi-exponential. Then the HJMM equation (\ref{HJMM-Wiener}) has an affine realization.
\end{corollary}

\begin{proof}
This follows from Theorem~\ref{thm-HJMM-Wiener} and Proposition~\ref{prop-suff-risk-neutral}.
\end{proof}

\begin{remark}\label{remark-geometric}
Let us present two geometric interpretations of Theorem~\ref{thm-HJMM-Wiener}. For this purpose, let $(\calm_t)_{t \geq 0}$ be a foliation generated by some finite dimensional subspace $V$, and suppose that this foliation is invariant for the HJMM equation (\ref{HJMM-Wiener}) with $y = y^*$.
\begin{enumerate}
\item By the tangential conditions (\ref{nu-pre-pre}), (\ref{sigma-pre-pre}) of Theorem~\ref{thm-foliation-1-pre}, for each $t \in \bbr_+$ we have
\begin{align*}
\frac{d}{d \xi} h - \sum_{k=1}^d \sigma^k(h) \Sigma^k(h) &\in T \calm_t, \quad h \in \calm_t,
\\ \sigma^k(h) &\in V, \quad h \in \calm_t \text{ and } k=1,\ldots,d.
\end{align*}
These two conditions imply that for every $y \in \caly$ we have
\begin{align*}
\frac{d}{d \xi} h - \sum_{k=1}^d \sigma^k(h) \Sigma^k(h) + \sum_{k=1}^d \Theta^k(y) \sigma^k(h) &\in T \calm_t, \quad h \in \calm_t,
\end{align*}
that is, the required tangential condition (\ref{nu-pre}) of Theorem~\ref{thm-foliation-1} regarding the drift term is fulfilled.

\item Let $y \in \caly$ be such that $\theta = \Theta(Y^y)$ satisfies the conditions of Proposition~\ref{prop-NA}. Then, by Propositions~\ref{prop-NA} and \ref{prop-NA-dynamics}, the dynamics of (\ref{HJMM-Wiener}) with drift $\alpha(\cdot,y)$ are obtained from the dynamics of (\ref{HJMM-Wiener}) with drift $\alpha(\cdot,y^*)$ after an equivalent change of measure. Consequently, the foliation $(\calm_t)_{t \geq 0}$ is also invariant for the HJMM equation (\ref{HJMM-Wiener}) with drift term $\alpha(\cdot,y)$.
\end{enumerate}
\end{remark}

\begin{remark}\label{remark-geometric-2}
We have seen in Theorem~\ref{thm-HJMM-affine-real} that the situation becomes more involved if, instead of the Wiener process driven HJMM equation (\ref{HJMM-Wiener}), we consider the general HJMM equation (\ref{HJMM}) with jumps. As soon as we have jumps, the two geometric interpretations from Remark~\ref{remark-geometric} fail, as we shall briefly explain:
\begin{enumerate}
\item The condition
\begin{align*}
\frac{d}{d\xi} h + \alpha(h,y^*) \in T \calm_t, \quad h \in \calm_t
\end{align*}
does generally not imply
\begin{align*}
\frac{d}{d\xi} h + \alpha(h,y) \in T \calm_t, \quad h \in \calm_t \text{ and } y \in \caly,
\end{align*}
because the drift term (\ref{alpha-HJM}) becomes too involved.

\item As we have seen in Remark~\ref{remark-measure-change}, for the general HJMM equation (\ref{HJMM}) with jumps, a change of the market price of risk can no longer (not even in a broad sense) be interpreted as an equivalent measure change.
\end{enumerate}
\end{remark}

\begin{example}
We consider the HJMM equation (\ref{HJMM-Wiener}) with a one-dimensional Wiener process $W$ and a quasi-exponential volatility $\sigma : H \to H$, and we choose the state space $\caly = \bbr_+$ with $y^* = 0$. For $y_0 \in (0,\infty)$ we denote by $Y^{y_0}$ the solution to the SDE (\ref{SDE-Z-Bessel}) provided in Example~\ref{example-squared-Bessel}. Furthermore, we define the mapping $\Theta : \caly \to \bbr$ as $\Theta(y) := 2 \sqrt{y}$. Then, according to Proposition~\ref{prop-suff-risk-neutral} and Corollary~\ref{cor-HJMM-Wiener}, the HJMM equation (\ref{HJMM-Wiener}) has an affine realization. Moreover, as seen in Example~\ref{example-squared-Bessel}, for no choice of the initial value $y_0 \in (0,\infty)$ the interest rate model admits an equivalent local martingale measure.
\end{example}

\section{Necessary conditions for the existence of an affine realization for the HJMM equation driven by a L\'{e}vy process and drift term described by its cumulant generating function}\label{sec-Levy-cumulant}

In this section, we derive necessary conditions for the existence of an affine realization for the HJMM equation, for which we assume a particular structure of the market price of risk regarding the jump part. For simplicity, we assume that the L\'{e}vy process $X$ in (\ref{HJMM}) is one-dimensional. We denote its L\'{e}vy measure by $F$, and suppose that the mapping $\Phi : \caly \times \bbr \to (0,\infty)$ is of the form
\begin{align*}
\Phi(y,x) = \exp(x \cdot \vartheta(y))
\end{align*}
with a continuous mapping $\vartheta : \caly \to \bbr$. Furthermore, we suppose that the topological space $\caly$ is connected, and that there exists $\epsilon > 0$ such that
\begin{align*}
\int_{\{|x| > 1\}} e^{zx} F(dx) < \infty \quad \text{for all $z \in (-\epsilon,\epsilon)$.}
\end{align*}
We define the cumulant generating function $\kappa : (-\epsilon,\epsilon) \to \bbr$ as follows. If $X$ is of type $X = x * \mu^X$, then we set
\begin{align*}
\kappa(z) := \int_{\bbr} (e^{zx} - 1) F(dx), \quad z \in (-\epsilon,\epsilon),
\end{align*}
and if $X$ is of type $X = x * (\mu^X - \nu)$, then we set
\begin{align*}
\kappa(z) := \int_{\bbr} (e^{zx} - 1 - zx) F(dx) \quad z \in (-\epsilon,\epsilon).
\end{align*}
The cumulant generating function $\kappa$ is real analytic on $(-\epsilon,\epsilon)$. We suppose that
\begin{align*}
\Gamma(h)(\xi) &\in (-\epsilon,\epsilon) \quad \text{for all $h \in H$ and $\xi \in \bbr_+$,}
\\ \vartheta(y) + \Gamma(h)(\xi) &\in (-\epsilon,\epsilon) \quad \text{for all $(y,h) \in \caly \times H$ and $\xi \in \bbr_+$.}
\end{align*}

\begin{proposition}\label{prop-cumulant}
We suppose that the HJMM equation (\ref{HJMM}) has an affine realization. Furthermore, we suppose that $\gamma \not\equiv 0$ and that
\begin{align}\label{kappa-infty}
\langle \kappa^{(m)} : m \in \bbn_0 \rangle = \infty.
\end{align}
Then $\vartheta$ is constant.
\end{proposition}

\begin{proof}
Suppose, on the contrary, that $\vartheta$ is not constant. Since $\caly$ is connected and $\vartheta$ is continuous, there exist $a,b \in \bbr$ with $a < b$ such that $[a,b] \subset \vartheta(\caly)$.
By the continuity of $\gamma$ there exists $h \in H$ such that $\gamma(h) \neq 0$, which implies $\Gamma(h) \neq 0$. For each $y \in \caly$ we have
\begin{align*}
&\int_{\bbr} \Psi(y,x) e^{x \Gamma(h)} F(dx) = \int_{\bbr} ( 1 - \Phi(y,x)) e^{x \Gamma(h)} F(dx)
\\ &= \int_{\bbr} \big( 1 - e^{x \vartheta(y)} \big) e^{x \Gamma(h)} F(dx) = \kappa(\Gamma(h)) - \kappa(\vartheta(y) + \Gamma(h)).
\end{align*}
Since $\Gamma(h) \neq 0$, by Proposition~\ref{prop-lin-ind-of-span-inf} and (\ref{kappa-infty}) we deduce that $U_{\Psi,\gamma}$ is infinite dimensional, which contradicts Theorem~\ref{thm-HJMM-affine-real}.
\end{proof}

Condition (\ref{kappa-infty}) means that for no $m \in \bbn$ the function $\kappa$ satisfies a linear ordinary differential equation of $m$th-order $\kappa^{(m)} = \sum_{k=0}^{m-1} c_k \kappa^{(k)}$. This condition is typically satisfied for pure jump L\'{e}vy processes with infinite activity, for example bilateral Gamma processes (see \cite{Kuechler-Tappe}), which cover the popular class of Variance Gamma processes, or for tempered stable processes (see \cite{Kuechler-Tappe-TS} and references therein), which cover the well-studied class of CGMY processes. Condition (\ref{kappa-infty}) is typically also satisfied for compound Poisson processes with infinite jump size distribution, but it is not satisfied for compound Poisson processes with finite jump size distribution, because then the cumulant generating function is of the form
\begin{align*}
\kappa(z) = c \bigg( \sum_{x \in \calx} \pi(x) e^{zx} - 1 \bigg)
\end{align*}
for some constant $c > 0$, a finite set $\calx$ and a stochastic vector $\pi : \calx \to (0,1]$. This is in accordance with our upcoming results; see, for example, Proposition~\ref{prop-Fc-Fd} below.

\section{Sufficient conditions for the existence of an affine realization for the HJMM equation}\label{sec-suff-prod}

In this section, we provide sufficient conditions for the existence of an affine realization of the HJMM equation (\ref{HJMM}). As we have seen in Theorem~\ref{thm-HJMM-affine-real}, we need that the subspace $U_{\Psi,\gamma}$ defined in (\ref{def-U-Psi-gamma}) is finite dimensional. We will establish two other conditions, which are easier to check, and which imply the finite dimensionality of $U_{\Psi,\gamma}$. 
For this purpose, note that, due to Assumption \ref{ass-Lipschitz-HJMM}, for each $k = 1,\ldots,n$ we can regard $\Psi^k$ and $x \mapsto e^{x \Gamma^k}$ as mappings
\begin{align}\label{Psi-regarded}
\Psi^k : \caly \to L^p(F^k) \quad \text{and} \quad e^{\Gamma^k} : H \to L^q(F^k;H).
\end{align}
We introduce the subspaces $U_{\Psi^k} \subset L^p(F^k)$ and $U_{\gamma^k} \subset L^q(F^k;H)$ for $k=1,\ldots,n$ by (\ref{def-U-Psi}) and (\ref{def-U-gamma}).

\begin{proposition}\label{prop-fin-dim-suff}
Suppose that the subspaces $U_{\Psi^1},\ldots,U_{\Psi^n}$ and $U_{\gamma^1},\ldots,U_{\gamma^n}$ are finite dimensional. Then the following statements are equivalent:
\begin{enumerate}
\item[(i)] The HJMM equation (\ref{HJMM}) has an affine realization.

\item[(ii)] The HJMM equation (\ref{HJMM}) for $Y_0 = y^*$ has an affine realization.
\end{enumerate}
\end{proposition}

\begin{proof}
We define the subspaces $U_k \subset L^1(F^k;H)$ as
\begin{align*}
U_k := \langle x \mapsto \Psi^k(x,y) e^{x \Gamma^k(h)} : h \in H \text{ and } y \in \caly \rangle, \quad k=1,\ldots,n.
\end{align*}
By assumption, we have
\begin{align*}
\dim U_k \leq \dim U_{\Psi^k} \cdot \dim U_{\gamma^k} < \infty, \quad k=1,\ldots,n.
\end{align*}
Denoting by $T^k : L^1(F^k;H) \to H$ the integral operator
\begin{align*}
T^k \psi := \int_{\bbr} \psi(x) F^k(dx), \quad k=1,\ldots,n,
\end{align*}
by the definition (\ref{def-U-Psi-gamma}) of the subspace $U_{\Psi,\gamma}$ we obtain
\begin{align*}
\dim U_{\Psi,\gamma} \leq \sum_{k=1}^n \dim T^k(U_k) \leq \sum_{k=1}^n \dim U_k < \infty,
\end{align*}
and hence, the claimed equivalence follows from Theorem~\ref{thm-HJMM-affine-real}.
\end{proof}

\begin{corollary}\label{cor-fin-dim-suff}
Suppose that $X^1,\ldots,X^n$ are compound Poisson processes with finite jump size distributions, and that $\gamma^1,\ldots,\gamma^n$ are constant.  Then the following statements are equivalent:
\begin{enumerate}
\item[(i)] The HJMM equation (\ref{HJMM}) has an affine realization.

\item[(ii)] The HJMM equation (\ref{HJMM}) for $Y_0 = y^*$ has an affine realization.
\end{enumerate}
\end{corollary}

\begin{proof}
This is a direct consequence of Proposition~\ref{prop-fin-dim-suff}.
\end{proof}

\begin{corollary}\label{cor-suff-FDR}
Suppose that the following conditions are satisfied:
\begin{enumerate}
\item $X^1,\ldots,X^n$ are compound Poisson processes with finite jump size distributions.

\item $\sigma^1,\ldots,\sigma^d$ are quasi-exponential.

\item $\gamma^1,\ldots,\gamma^n$ are constant and quasi-exponential.
\end{enumerate}
Then the HJMM equation (\ref{HJMM}) has an affine realization.
\end{corollary}

\begin{proof}
This follows by combining Proposition~\ref{prop-suff-risk-neutral} and Corollary~\ref{cor-fin-dim-suff}.
\end{proof}

\begin{example}
We consider the HJMM equation (\ref{HJMM}) with a one-dimensional Wiener process $W$ and a one-dimensional standard Poisson process $X$. Let $\sigma : H \to H$ be quasi-exponential, and let $\gamma : H \to H$ be constant and quasi-exponential. We choose the state space $\caly = \bbr_+$ with $y^* = 0$. For $y_0 \in (0,\infty)$ we denote by $Y^{y_0}$ the solution to the SDE (\ref{SDE-Z-Bessel}) provided in Example~\ref{example-squared-Bessel}. Furthermore, we define the mapping $\Theta : \caly \to \bbr$ as $\Theta(y) := 2 \sqrt{y}$, and we define the mapping $\Psi : \caly \times \bbr \to(-\infty,1)$ as $\Psi(y,x) := -y$. Then, according to Corollary~\ref{cor-suff-FDR}, the HJMM equation (\ref{HJMM}) has an affine realization. Moreover, as seen in Example~\ref{example-squared-Bessel-2}, for no choice of the initial value $y_0 \in (0,\infty)$ the interest rate model admits an equivalent local martingale measure.
\end{example}

Now, it arises the question whether the finite dimensionality of the  subspaces $U_{\Psi^1},\ldots,U_{\Psi^n}$ and $U_{\gamma^1},\ldots,U_{\gamma^n}$ is also necessary for the existence of an affine realization. We will deal with this question in the upcoming two sections.

\section{Necessary conditions on the market price of risk for the existence of an affine realization}\label{sec-nec-Psi}

In this section, we investigate the necessity of the first assumption from Proposition~\ref{prop-fin-dim-suff}. More precisely, we investigate whether the subspaces generated by the market price of risk must necessarily be finite dimensional for the existence of an affine realization.

For simplicity, we assume that the L\'{e}vy process $X$ in (\ref{HJMM}) is one-dimensional, and denote its L\'{e}vy measure by $F$. In addition to the assumptions from Section \ref{sec-affine-real-general}, we suppose that the market price of risk can even be regarded as a mapping
\begin{align*}
\Psi : \caly \to L^1(F) \cap L^p(F),
\end{align*}
that is, it should not only map into $L^p(F)$, as stated in (\ref{Psi-regarded}), but also into $L^1(F)$,

\begin{theorem}\label{thm-dim-U-Psi}
Suppose that the HJMM equation (\ref{HJMM}) has an affine realization, and that one of the following conditions is satisfied:
\begin{enumerate}
\item[($\call$)] ${\rm supp}(F) \subset \bbr_+$ and $\bbr_- \subset  \Gamma(h)(\bbr_+)$ for some $h \in H$.

\item[($\call^{\epsilon}$)] There exists $\epsilon > 0$ such that $(-\epsilon,\epsilon) \subset \Gamma(h)(\bbr_+)$ for some $h \in H$. 
\end{enumerate}
Then the subspace $U_{\Psi}$ is finite dimensional.
\end{theorem}

\begin{proof}
We choose $h \in H$ such that condition ($\call$) or ($\call^{\epsilon}$) is fulfilled. We define the integral operator
\begin{align*}
T_h : L^1(F) \cap L^p(F) \to H, \quad T_h \psi := \int_{\bbr} \psi(x) e^{x \Gamma(h)} F(dx),
\end{align*}
and claim that ${\rm ker}(T_h) = \{ 0 \}$. Indeed, let $\psi \in L^1(F) \cap L^p(F)$ be such that $T_h \psi = 0$. Then we have
\begin{align}\label{T-y-plus-minus}
T_h \psi^+ = T_h \psi^-.
\end{align}
Since $\psi \in L^1(F)$, we can define the finite signed measure $\mu$ on $(\bbr,\calb(\bbr))$ by $\frac{d \mu}{d F} := \psi$. Its Jordan decomposition $\mu = \mu^+ - \mu^-$ is given by $\frac{d \mu^+}{d F} = y^+$ and $\frac{d \mu^-}{d F} = y^-$. Now, we distinguish the two cases from our hypothesis of the theorem:
\begin{enumerate}
\item[($\call$)] Since, ${\rm supp}(F) \subset \bbr_+$, the measures $\mu^+$ and $\mu^-$ are finite measures on $(\bbr_+,\calb(\bbr_+))$. We denote by $\call_{\mu^+},\call_{\mu^-} : \bbr_+ \to \bbr_+$ their Laplace transforms. Let $\lambda \in \bbr_+$ be arbitrary. By assumption, there exists $\xi \in \bbr_+$ such that $\Gamma(h)(\xi) = -\lambda$. Therefore, we have
\begin{align*}
\call_{\mu^+}(\lambda) &= \int_{\bbr_+} e^{- \lambda x} \mu^+(dx) = \int_{\bbr_+} e^{x \Gamma(h)(\xi)} \mu^+(dx) 
\\ &= \int_{\bbr_+} \psi^+(x) e^{x \Gamma(h)(\xi)} F(dx) = T_h \psi^+,
\end{align*}
and an analogous calculation shows that
\begin{align*}
\call_{\mu^-}(\lambda) = T_h \psi^-.
\end{align*}
In view of (\ref{T-y-plus-minus}), we deduce that
\begin{align*}
\call_{\mu^+}(\lambda) = \call_{\mu^-}(\lambda) \quad \text{for all $\lambda \in \bbr_+$.}
\end{align*}
By the first uniqueness theorem for one-sided Laplace transforms (Theorem~\ref{thm-Laplace-plus}), we deduce that $\mu^+ = \mu^-$.

\item[($\call^{\epsilon}$)] We denote by $\call_{\mu^+}^{\epsilon},\call_{\mu^-}^{\epsilon} : (-\epsilon,\epsilon) \to \bbr_+$ the Laplace transforms of $\mu^+$ and $\mu^-$. Let $\lambda \in (-\epsilon,\epsilon)$ be arbitrary. By assumption, there exists $\xi \in \bbr_+$ such that $\Gamma(h)(\xi) = -\lambda$. Therefore, an analogous calculation as in (i) shows that
\begin{align*}
\call_{\mu^+}^{\epsilon}(\lambda) = \call_{\mu^-}^{\epsilon}(\lambda) \quad \text{for all $\lambda \in (-\epsilon,\epsilon)$.}
\end{align*}
By the second uniqueness theorem for two-sided Laplace transforms (Theorem~\ref{thm-Laplace-epsilon}), we deduce that $\mu^+ = \mu^-$.  
\end{enumerate}
Consequently, in both cases, we deduce that $\psi^+ = \psi^-$ almost surely with respect to $F$, which implies $\psi = 0$ almost surely with respect to $F$. This proves ${\rm ker}(T_h) = \{ 0 \}$.
By Theorem~\ref{thm-HJMM-affine-real}, the range $T_h(U_{\Psi}) \subset U_{\Psi,\gamma}$ is finite dimensional, and hence, we deduce that $U_{\Psi}$ is finite dimensional, too.
\end{proof}

For the rest of this section we suppose that $\Phi : \caly \times \bbr \to (0,\infty)$ is of the form
\begin{align}\label{exp-prod}
\Phi(y,x) = \exp(\vartheta(y) \xi(x))
\end{align}
with a continuous mapping $\vartheta : \caly \to \bbr$ and a measurable mapping $\xi : \bbr \to \bbr$. We suppose that $\caly$ is connected. In the sequel, we denote by $F^d$ the discrete part of the L\'{e}vy measure $F$, and by $F^c$ its absolutely continuous part. 

\begin{proposition}\label{prop-Fc-Fd}
Suppose that the HJMM equation (\ref{HJMM}) has an affine realization, and that one of conditions ($\call$) or ($\call^{\epsilon}$) from Theorem~\ref{thm-dim-U-Psi} is fulfilled. Furthermore, suppose that one of the following conditions is satisfied:
\begin{enumerate}
\item[($F^d$)] $\xi({\rm supp}(F^d))$ is infinite.

\item[($F^c$)] There are $c,d \in \bbr$ with $c < d$ such that $\xi$ is continuous on $[c,d]$ with $\xi(c) \neq \xi(d)$ and $\frac{dF^c}{d\lambda}(x) > 0$ for $F^c$-almost all $x \in [c,d]$.
\end{enumerate}
Then $\vartheta$ is constant.
\end{proposition}

\begin{proof}
Suppose that $\vartheta$ is not constant. Since $\caly$ is connected and $\vartheta$ is continuous, there exist $a,b \in \bbr$ with $a < b$ such that $[a,b] \subset \vartheta(\caly)$. Now, we distinguish the two cases from our hypothesis of the theorem:
\begin{itemize}
\item[($F^d$)] We set $\calx := {\rm supp}(F^d)$ and introduce the linear operator
\begin{align*}
\Pi : U_{\Psi} \to \ell(\calx), \quad \psi \mapsto \psi|_{\calx},
\end{align*}
where we note that for all $\psi,\phi \in \call^p(F)$ with the same equivalence class $[\psi] = [\phi]$ we have $\psi|_{\calx} = \phi|_{\calx}$. Since 
\begin{align*}
[\![ a \cdot \psi|_{\calx},b \cdot \psi|_{\calx} ]\!] \subset \{ \vartheta(y) \cdot \psi|_{\calx} : y \in \caly \}
\end{align*}
and the set $\xi(\calx)$ is infinite, by Proposition~\ref{prop-lin-ind-exp} the subspace
\begin{align*}
\Pi(U_{\Psi}) = \langle \exp(\vartheta(y) \cdot \xi|_{\calx}) - 1 : y \in \caly \rangle 
\end{align*}
is an infinite dimensional subspace of $\ell(\calx)$.

\item[($F^c$)] We define the interval $I = [c,d]$ and the linear operator
\begin{align*}
\Pi : U_{\Psi} \to \ell(I), \quad \psi \mapsto \psi|_I,
\end{align*}
where for $\psi \in L^p(F)$ we define $\psi|_I$ as the unique continuous representative of $\psi$ on $I$ according to Lemma~\ref{lemma-one-repr}. Since 
\begin{align*}
[\![ a \cdot \psi|_{I},b \cdot \psi|_{I} ]\!] \subset \{ \vartheta(y) \cdot \psi|_{I} : y \in \caly \}
\end{align*}
the set $\xi(I)$ is infinite, by Proposition~\ref{prop-lin-ind-exp} the subspace
\begin{align*}
\Pi(U_{\Psi}) = \langle \exp(\vartheta(y) \cdot \xi|_{I}) - 1 : y \in \caly \rangle 
\end{align*}
is an infinite dimensional subspace of $\ell(I)$.
\end{itemize}
Consequently, in both cases we deduce that $U_{\Psi}$ is infinite dimensional, which contradicts Theorem~\ref{thm-dim-U-Psi}.
\end{proof}

Note that the hypotheses of Proposition~\ref{prop-Fc-Fd} regarding the L\'{e}vy measure $F$ and the function $\xi$ are, in particular, satisfied in the following situations:
\begin{itemize}
\item $X$ is a compound Poisson process with infinite jump size distribution, and $\xi$ is one-to-one on the support of $F$.

\item The L\'{e}vy measure of $X$ has a strictly positive density on some subinterval of positive length, and $\xi$ is continuous and nontrivial on this subinterval. Most of the infinite activity L\'{e}vy processes, which are considered in the literature, have a strictly positive density, for example bilateral Gamma processes or tempered stable processes, which we have mentioned in Section \ref{sec-Levy-cumulant}.
\end{itemize}
In this sense, Proposition~\ref{prop-Fc-Fd} generalizes Proposition~\ref{prop-cumulant}, where we have considered the identity mapping $\xi(x) = x$. On the other hand, for Proposition~\ref{prop-Fc-Fd} we have imposed that one of the conditions ($\call$) or ($\call^{\epsilon}$) is fulfilled.

\section{Necessary conditions on the volatility for the existence of an affine realization}\label{sec-nec-gamma}

In this section, we investigate the necessity of the second assumption from Proposition~\ref{prop-fin-dim-suff}. More precisely, we investigate whether the subspaces generated by the volatility must necessarily be finite dimensional for the existence of an affine realization.

For simplicity, we assume that the L\'{e}vy process $X$ in (\ref{HJMM}) is a compound Poisson process with finite jump size distribution. Then the first condition from Proposition~\ref{prop-fin-dim-suff} is fulfilled; see also Corollary~\ref{cor-fin-dim-suff}.

\begin{theorem}\label{thm-convex}
Suppose that the HJMM equation (\ref{HJMM}) has an affine realization. Furthermore, suppose that $\Psi \not\equiv 0$ and that $[\![ 0,g ]\!] \subset \gamma(H)$ for some $g \in H$, where the line segment $[\![ 0,g ]\!]$ is defined as
\begin{align*}
[\![ 0,g ]\!] := \{ t g : t \in [0,1] \}.
\end{align*}
Then $\gamma$ is constant, and in particular $U_{\gamma}$ is finite dimensional
\end{theorem}

\begin{proof}
Suppose, on the contrary, that $\gamma$ is not constant. Then we have $[\![ 0,f ]\!] \subset \Gamma(H)$, where $f := \int_0^{\bullet} g(\eta) d\eta$. Since $X$ is a compound Poisson process with finite jump size distribution, there exist a finite set $\calx \subset \bbr \setminus \{ 0 \}$ and a mapping $\rho : \calx \to (0,\infty)$ such that the L\'{e}vy measure $F$ of $X$ is given by
\begin{align*}
F(B) = \sum_{x \in \calx \cap B} \rho(x) \quad \text{for all $B \in \calb(\bbr)$.} 
\end{align*}
Since $\Psi \not\equiv 0$, there exists $y \in \caly$ such that $\Psi(y,x) \neq 0$ for some $x \in \calx$. By Theorem~\ref{thm-HJMM-affine-real} the subspace $U \subset U_{\Psi,\gamma}$ given by
\begin{align*}
U := \Big\langle \sum_{x \in \calx} \rho(x) \Psi(y,x) e^{x \Gamma(h)} : h \in H \Big\rangle,
\end{align*}
is finite dimensional. Now, let $m \in \bbn$ be arbitrary. First, we show that there exist $t_1,\ldots,t_m \in [0,1]$ such that the elements
\begin{align*}
x \cdot t_i, \quad x \in \calx \text{ and } i=1,\ldots,m
\end{align*}
are pairwise different. Indeed, by induction we prove that for each $d = 1,\ldots,m$ there are $t_1,\ldots,t_d \in [0,1]$ such that the elements
\begin{align*}
x \cdot t_i, \quad x \in \calx \text{ and } i=1,\ldots,d
\end{align*}
are pairwise different. For $d=1$ we can choose $t_1 := 1$. For the induction step $d \to d+1$ we choose $t_{d+1} \in [0,1]$ such that the finitely many conditions
\begin{align*}
t_{d+1} \neq \frac{x_2 \cdot t_j}{x_1} \quad \text{for all $x_1,x_2 \in \calx$ and $j = 1,\ldots,d$,}
\end{align*}
are fulfilled. Then we have
\begin{align*}
x_1 \cdot t_{d+1} \neq x_2 \cdot t_j \quad \text{for all $x_1,x_2 \in \calx$ and $j = 1,\ldots,d$,}
\end{align*}
and hence the elements
\begin{align*}
x \cdot t_j, \quad x \in \calx \text{ and } j=1,\ldots,d+1
\end{align*}
are pairwise different. Now, let $c_1,\ldots,c_m \in \bbr$ be such that
\begin{align*}
\sum_{i=1}^m c_i \sum_{x \in \calx} \rho(x) \Psi(y,x) \exp(x \cdot t_i f) = 0.
\end{align*}
Then we have
\begin{align*}
\sum_{i=1}^m \sum_{x \in \calx} \big( c_i \rho(x) \Psi(y,x) \big) \exp(x \cdot t_i f) = 0.
\end{align*}
By Proposition~\ref{prop-lin-ind-exp-ps} we deduce that
\begin{align*}
c_i \rho(x) \Psi(y,x) = 0 \quad \text{for all $i=1,\ldots,m$ and all $x \in \calx$.}
\end{align*}
Since $\rho(x) > 0$ for all $x \in \calx$ and $\Psi(y,x) \neq 0$ for some $x \in \calx$, we deduce that $c_1 = \ldots = c_m = 0$. Since $m \in \bbn$ was arbitrary, and we have $[\![ 0,f ]\!] \subset \Gamma(H)$, we arrive at the contradiction that the subspace $U_{\Psi,\gamma}$ is infinite dimensional.
\end{proof}

\section{Conclusion}\label{sec-conclusion}

In this paper, we have investigated the existence of affine realization for the L\'{e}vy process driven HJMM equation (\ref{HJMM}) with real-world forward rate dynamics. To sum up our findings of the previous sections, we have seen that, under suitable conditions, the HJMM equation (\ref{HJMM}) has an affine realization if and only if the following three conditions are satisfied:
\begin{enumerate}
\item[(i)] The risk-neutral HJMM equation has an affine realization.

\item[(ii)] The subspaces $U_{\Psi^1},\ldots,U_{\Psi^n}$ are finite dimensional.

\item[(iii)] The subspaces $U_{\gamma^1},\ldots,U_{\gamma^n}$ are finite dimensional.
\end{enumerate}
Thus, if one is interested in the existence of an affine realization, this suggests to choose models where $\gamma^1,\ldots,\gamma^n$ are constant and the pure jump L\'{e}vy processes $X^1,\ldots,X^n$ are compound Poisson processes with finite jump size distributions. In this case, a sufficient condition for the existence of an affine realization is that all volatilities $\sigma^1,\ldots,\sigma^d$ and $\gamma^1,\ldots,\gamma^n$ are quasi-exponential.

Consequently, when considering the full picture of the term structure dynamics under their real-world constraints, in the case of infinite activity L\'{e}vy process driven dynamics only rather restricted term structure models remain possible. This has obvious consequences for realistic term structure modeling.

\section*{Acknowledgement}

We are grateful to an anonymous referee for the careful study of our paper and the valuable comments and suggestions.

\begin{appendix}

\section{Results about real analytic functions}\label{app-analytic}

In this appendix we provide results about real analytic functions, which we require in this article. In the following, we denote by $J \subset I \subset \mathbb{R}$ two arbitrary nonempty, open intervals.

\begin{lemma}\label{lemma-lin-ind}
Let $f_1,\ldots,f_m : I \rightarrow \mathbb{R}$ be linearly independent, real analytic functions for some $m \in \mathbb{N}$. Then, there exist elements $\theta_1,\ldots,\theta_m \in J$ such that
\begin{align*}
\det \left(
\begin{array}{ccc}
f_1(\theta_1) & \cdots & f_1(\theta_m)
\\ \vdots & \ddots & \vdots
\\ f_m(\theta_1) & \cdots & f_m(\theta_m)
\end{array}
\right) \neq 0.
\end{align*}
\end{lemma}

\begin{proof}
First, we suppose that $m=1$. Since $f_1 \not \equiv 0$ on $I$ by the assumed linear independence, according to the identity theorem for analytic functions there exists $\theta_1 \in J$ with $f_1(\theta_1) \neq 0$. 
For $m \geq 2$ we proceed by induction and suppose there exist $\theta_1,\ldots,\theta_{m-1} \in J$ such that
\begin{align*}
\det \left(
\begin{array}{ccc}
f_1(\theta_1) & \cdots & f_1(\theta_{m-1})
\\ \vdots & \ddots & \vdots
\\ f_{m-1}(\theta_1) & \cdots & f_{m-1}(\theta_{m-1})
\end{array}
\right) \neq 0.
\end{align*}
Then the function $g : I \rightarrow \mathbb{R}$ given by
\begin{align*}
g(\theta) = \det \left(
\begin{array}{cccc}
f_1(\theta_1) & \cdots & f_1(\theta_{m-1}) & f_1(\theta)
\\ \vdots & \ddots & \vdots & \vdots
\\ f_{m-1}(\theta_1) & \cdots & f_{m-1}(\theta_{m-1}) & f_{m-1}(\theta)
\\ f_m(\theta_1) & \cdots & f_m(\theta_{m-1}) & f_m(\theta)
\end{array}
\right)
\end{align*}
is also real analytic, and it is of the form
\begin{align*}
g(\theta) = \sum_{i=1}^m \xi_i f_i(\theta), \quad \theta \in I
\end{align*}
with $\xi_1,\ldots,\xi_m \in \mathbb{R}$ and $\xi_m \neq 0$. Since $f_1,\ldots,f_m$ are linearly independent, we have $g \not \equiv 0$ on $I$. Since $g$ is real analytic, by the identity theorem for analytic functions there exists $\theta_m \in J$ with $g(\theta_m) \neq 0$, which finishes the proof.
\end{proof}

\begin{remark}
We refer to \cite[Prop.~5.5]{BKR0} for a result which has similarities to Lemma~\ref{lemma-lin-ind}.
\end{remark}

\begin{proposition}\label{prop-lin-ind-of-span-inf}
Let $f : I \rightarrow \mathbb{R}$ be a real analytic function and let $\Lambda : \mathbb{R}_+ \rightarrow \mathbb{R}$ be a continuous, non-constant function with $\Lambda(x) = 0$ for some $x \in \mathbb{R}_+$ such that
\begin{align*}
\theta + \Lambda(x) \in I \quad \text{for all $\theta \in J$ and $x \in \mathbb{R}_+$.}
\end{align*}
If we have
\begin{align}\label{dim-f-infinite}
\dim \langle f^{(n)} : n \in \mathbb{N}_0 \rangle = \infty,
\end{align}
then we also have
\begin{align}\label{dim-f-Lambda-infinite}
\dim \langle f(\theta + \Lambda) : \theta \in J \rangle = \infty.
\end{align}
\end{proposition}

\begin{proof}
The proof has a certain similarity to the one in \cite[Thm.~7.1]{Tappe-Levy}. Let $m \in \mathbb{N}$ be arbitrary. By (\ref{dim-f-infinite}), the functions $f, f', \ldots, f^{(m-1)}$ are linearly independent.
Hence, by Lemma~\ref{lemma-lin-ind} there exist elements $\theta_1,\ldots,\theta_m \in J$ such that $\det B \neq 0$, where $B \in \mathbb{R}^{m \times m}$ denotes the matrix with $B_{ki} = f^{(k)}(\theta_i)$ for $k=0,\ldots,m-1$ and $i=1,\ldots,m$. We will show that
\begin{align}\label{dim-m}
\dim \langle f(\theta_i + \Lambda) : i=1,\ldots,m \rangle = m
\end{align}
Indeed, let $\xi_1,\ldots,\xi_m \in \mathbb{R}$ be such that
\begin{align*}
\sum_{i=1}^m \xi_i f(\theta_i + \Lambda) = 0.
\end{align*}
Since $f$ is real analytic on $I$, there exists $\epsilon > 0$ such that
\begin{align*}
\sum_{i=1}^m \xi_i f(\theta_i + z) &= \sum_{i=1}^m \xi_i \sum_{n=0}^{\infty} \frac{f^{(n)}(\theta_i)}{n!} z^n 
\\ &= \sum_{n=0}^{\infty} \bigg( \sum_{i=1}^m \xi_i \frac{f^{(n)}(\theta_i)}{n!} \bigg) z^n \quad \text{for all $z \in (-\epsilon,\epsilon)$.}
\end{align*}
This gives us
\begin{align*}
\sum_{n=0}^{\infty} \bigg( \sum_{i=1}^m \xi_i \frac{f^{(n)}(\theta_i)}{n!} \bigg) z^n = 0 \quad \text{for all $z \in \Lambda(\mathbb{R}_+) \cap (-\epsilon,\epsilon)$.}
\end{align*}
Since $\Lambda$ is continuous and non-constant with $\Lambda(x) = 0$ for some $x \in \mathbb{R}_+$, there exists a sequence $(z_n)_{n \in \mathbb{N}} \subset \Lambda(\mathbb{R}_+) \cap (-\epsilon,\epsilon)$ with $z_n \neq 0$, $n \in \mathbb{N}$ and $z_n \rightarrow 0$. Therefore, the identity theorem for power series applies and yields
\begin{align*}
\sum_{i=1}^{m} \xi_i f^{(n)}(\theta_i) = 0 \quad \text{for all $n \in \mathbb{N}$,}
\end{align*}
and it follows that $B \xi = 0$. Since $\det B \neq 0$, we deduce that $\xi_1,\ldots,\xi_m = 0$, which proves (\ref{dim-m}). Since $m \in \mathbb{N}$ was arbitrary, we conclude (\ref{dim-f-Lambda-infinite}), which finishes the proof.
\end{proof}

\section{Results about measures}\label{app-measures}

In this appendix we provide results about measures which we require in this paper; in particular uniqueness results about Laplace transforms. As the uniqueness theorem for two-sided Laplace transforms (see Theorem~\ref{thm-Laplace-epsilon}) was not immediately available in the literature, we provide a self-contained proof.

\begin{definition}
For a finite measure $\mu$ on $(\bbr^d,\calb(\bbr^d))$ we define its Fourier transform
\begin{align*}
\calf_{\mu} : \bbr^d \to \bbc, \quad \calf_{\mu}(u) := \int_{\bbr} e^{i \langle u,x\rangle} \mu(dx).
\end{align*}
\end{definition}

\begin{theorem}\label{thm-Fourier}
Let $\mu$ and $\nu$ be two finite measure on $(\bbr^d,\calb(\bbr^d))$ such that $\calf_{\mu} = \calf_{\nu}$. Then we have $\mu = \nu$.
\end{theorem}

\begin{proof}
See, for example, \cite[Satz~15.6]{Klenke}.
\end{proof}

\begin{definition}
For a finite measure $\mu$ on $(\bbr_+,\calb(\bbr_+))$ we define its Laplace transform
\begin{align*}
\call_{\mu} : \bbr_+ \to \bbr_+, \quad \call_{\mu}(\lambda) := \int_{\bbr_+} e^{-\lambda x} \mu(dx).
\end{align*}
\end{definition}

\begin{theorem}\label{thm-Laplace-plus}
Let $\mu$ and $\nu$ be two finite measure on $(\bbr_+,\calb(\bbr_+))$ such that $\call_{\mu} = \call_{\nu}$. Then we have $\mu = \nu$.
\end{theorem}

\begin{proof}
See, for example, \cite[Satz~15.6]{Klenke}.
\end{proof}

\begin{definition}
Let $\mu$ be a finite measure on $(\bbr,\calb(\bbr))$, and let $\epsilon > 0$ such that
\begin{align}\label{ex-auf-streifen}
\int_{\bbr} e^{-\lambda x} \mu(dx) < \infty \quad \text{for all $\lambda \in (-\epsilon,\epsilon)$.}
\end{align}
Then we define its Laplace transform
\begin{align*}
\call_{\mu}^{\epsilon} : (-\epsilon,\epsilon) \to \bbr_+, \quad \call_{\mu}^{\epsilon}(\lambda) := \int_{\bbr} e^{-\lambda x} \mu(dx).
\end{align*}
\end{definition}

\begin{lemma}\label{lemma-holomorphic}
Let $\mu$ be a finite measure on $(\bbr,\calb(\bbr))$, let $\epsilon > 0$ such that (\ref{ex-auf-streifen}) is satisfied, and let $G \subset \bbc$ be the open set
\begin{align}\label{def-G}
G := \{ z \in \bbc : {\rm Re} \, z \in (-\epsilon,\epsilon) \}.
\end{align}
Then the function
\begin{align}\label{def-Laplace-two}
L_{\mu} : G \to \bbc, \quad L_{\mu}(z) := \int_{\bbr} e^{-zx} \mu(dx)
\end{align}
is holomorphic.
\end{lemma}

\begin{proof}
We define the mapping
\begin{align*}
f : G \times \bbr \to \bbc, \quad f(z,x) := e^{-zx}.
\end{align*}
Then we easily verify that the following conditions are fulfilled:
\begin{enumerate}
\item[(a)] $f(z,\cdot) \in L^1$ for all $z \in G$.

\item[(b)] For all $x \in \bbr$ the mapping $f(\cdot,x) : G \to \bbc$ is holomorphic.

\item[(c)] For each compact ball $K \subset G$ there is a nonnegative function $g_K \in L^1$ such that $|f(z,\cdot)| \leq g_K$ for all $z \in K$.
\end{enumerate}
Therefore, the function $L_{\mu}$ is holomorphic by virtue of \cite[Satz~IV.5.8]{Elstrodt}.
\end{proof}

\begin{theorem}\label{thm-Laplace-epsilon}
Let $\mu$ and $\nu$ be two finite measures on $(\bbr,\calb(\bbr))$, and let $\epsilon > 0$ such that
\begin{align*}
\int_{\bbr} e^{-\lambda x} \mu(dx) < \infty \quad \text{and} \quad \int_{\bbr} e^{-\lambda x} \nu(dx) < \infty \quad \text{for all $\lambda \in (-\epsilon,\epsilon)$,}
\end{align*}
and $\call_{\mu}^{\epsilon} = \call_{\nu}^{\epsilon}$. Then we have $\mu = \nu$.
\end{theorem}

\begin{proof}
We define the open set $G \subset \bbc$ by (\ref{def-G}) and the functions $L_{\mu},L_{\nu} : G \to \bbc$ according to (\ref{def-Laplace-two}). Then $L_{\mu}$ and $L_{\nu}$ are holomorphic by Lemma~\ref{lemma-holomorphic}. Since $\call_{\mu}^{\epsilon} = \call_{\nu}^{\epsilon}$, we have $L_{\mu}|_H = L_{\nu}|_H$, where $H \subset G$ denotes the subset
\begin{align*}
H := \{ z \in \bbc : {\rm Re} \, z \in (-\epsilon,\epsilon) \text{ and } {\rm Im} \, z = 0 \}.
\end{align*}
By the identity theorem for holomorphic functions (see, for example, \cite[Satz~8.1.3]{Remmert}) we deduce that $L_{\mu} = L_{\nu}$. For all $u \in \bbr$ we have $iu \in G$, and hence
\begin{align*}
\calf_{\mu}(u) = \int_{\bbr} e^{iux} \mu(dx) = L_{\mu}(iu) = L_{\nu}(iu) = \int_{\bbr} e^{iux} \nu(dx) = \calf_{\nu}(u),
\end{align*}
showing that $\calf_{\mu} = \calf_{\nu}$. By Theorem~\ref{thm-Fourier}, we deduce that $\mu = \nu$.
\end{proof}

Our last auxiliary result of this appendix states that an equivalence class of a Lebesgue space can have at most one continuous representative.

\begin{lemma}\label{lemma-one-repr}
Let $I = [c,d]$ be an interval with $c,d \in \bbr$ and $c < d$, and let $F$ be an absolutely continuous measure on $(\bbr,\calb(\bbr))$ such that
\begin{align*}
\frac{dF}{d\lambda}(x) > 0 \quad \text{for $F$-almost all $x \in I$.} 
\end{align*}
Let $f,g : \bbr \to \bbr$ be two functions such that $f|_I$ and $g|_I$ are continuous, and we have
\begin{align}\label{f-gleich-g-as}
f(x) = g(x) \quad \text{for $\lambda$-almost all $x \in \bbr$.}
\end{align}
Then we have $f|_I = g|_I$.
\end{lemma}

\begin{proof}
By the continuity of $f$ and $g$, it suffices to prove that $f(x) = g(x)$ for all $x \in (c,d)$. Suppose, on the contrary, there exists $x \in (c,d)$ such that $f(x) \neq g(x)$. By the continuity of $f$ and $g$, there exists $\delta > 0$ such that $(x-\delta,x+\delta) \subset (c,d)$ and $f(y) \neq g(y)$ for all $y \in (x-\delta,x+\delta)$. Then, setting $\rho := \frac{dF}{d\lambda}$, we have
\begin{align*}
F((x-\delta,x+\delta)) = \int_{x-\delta}^{x+\delta} \rho(x) dx > 0,
\end{align*}
which contradicts (\ref{f-gleich-g-as}).
\end{proof}

\section{Results about linearly independent functions}\label{app-lin-ind}

In this appendix, we collect results about linearly independent functions.
Let $\calx$ be an infinite set. We denote by $\ell(\calx)$ the vector space of all functions $f : \calx \to \bbr$. For $f,g \in \ell(\calx)$ we define the line segment $[\![ f,g ]\!] \subset \ell(\calx)$ as
\begin{align*}
[\![ f,g ]\!] := \{ f + t(g - f) : t \in [0,1] \}.
\end{align*}

\begin{proposition}\label{prop-lin-ind-exp}
Let $f,g \in \ell(\calx)$ be such that the set $(g-f)(\calx)$ is infinite. Then the following statements are true:
\begin{enumerate}
\item The subspace $U \subset \ell(X)$ given by
\begin{align*}
U := \langle \exp(h) : h \in [\![ f,g ]\!] \rangle
\end{align*}
is infinite dimensional.

\item If $[\![ f,g ]\!] \subset \langle g-f \rangle$, then the subspace $V \subset \ell(X)$ given by
\begin{align*}
V := \langle \exp(h) - 1 : h \in [\![ f,g ]\!] \rangle
\end{align*}
is infinite dimensional.
\end{enumerate}
\end{proposition}

\begin{proof}
We set $h := g-f$. First, we will show that for each $m \in \bbn$ the functions
\begin{align}\label{lin-ind-dyadic}
\exp(j 2^{-m} h), \quad j=-2^m,\ldots,2^m
\end{align}
are linearly independent in $\ell(X)$. Indeed, let $c_{-2^m},\ldots,c_{2^{m}} \in \bbr$ be such that
\begin{align*}
\sum_{k=-2^m}^{2^m} c_j \exp(j 2^{-m} h) = 0.
\end{align*}
Defining $d := 2^{m+1}$ and the vector $\gamma \in \bbr^{d+1}$ as $\gamma_j := c_{j - 2^m}$ for $j=0,\ldots,d$ we obtain
\begin{align*}
\exp(2^{-m} h) \sum_{j=0}^{d} \gamma_j \exp(2^{-m} h)^j = 0.
\end{align*}
Since $h(\calx)$ is infinite by assumption, there exist elements $x_0,\ldots,x_d \in \calx$ such that $h(x_i)$, $i=0,\ldots,d$ are pairwise different. We obtain
\begin{align*}
\sum_{j=0}^{d} \gamma_j \exp(2^{-m} h(x_i))^j = 0, \quad i=1,\ldots,d.
\end{align*}
Defining the Vandermonde matrix $A \in \bbr^{(d+1) \times (d+1)}$ as $A_{ij} := \exp(2^{-m} h(x_i))^j$ for $i,j=0,\ldots,d$ we obtain $A \cdot \gamma = 0$. Since $\exp(2^{-m} h(x_i))$, $i=0,\ldots,d$ are pairwise different, we deduce that $\gamma = 0$, and hence $c_{-2^m} = \ldots = c_{2^{m}} = 0$, showing the linear independence of the functions (\ref{lin-ind-dyadic}). Now, we are ready to prove the two statements:
\begin{enumerate}
\item Let $m \in \bbn$ be arbitrary, and let $c_{0},\ldots,c_{2^{m}} \in \bbr$ be such that
\begin{align*}
\sum_{j = 0}^{2^m} c_j \exp(f + j 2^{-m} h) = 0.
\end{align*}
Then we have
\begin{align*}
\exp(f) \sum_{j = 0}^{2^m} c_j \exp(j 2^{-m} h) = 0,
\end{align*}
and hence $c_0 = \ldots = c_{2^m} = 0$ by the linear independence of the functions (\ref{lin-ind-dyadic}). Since $m \in \bbn$ was arbitrary, we deduce that subspace $U$ is infinite dimensional.

\item Let $m \in \bbn$ be arbitrary. Since $[\![ f,g ]\!] \subset \langle h \rangle$, there exist $\lambda \in \bbr \setminus \{ 0 \}$ and $n \in \bbn$ such that the set
\begin{align*}
[\![ f,g ]\!] \cap \{ j 2^{-n} \lambda h : j = 1,\ldots,2^n \}
\end{align*}
has at least $m$ elements. Hence, by the linear independence of the functions (\ref{lin-ind-dyadic}) there are $h_1,\ldots,h_n \in [\![ f,g ]\!] \setminus \{ 0 \}$ such that with $h_0 := 0$ the functions
\begin{align}\label{lin-ind-line}
\exp(h_i), \quad i=0,\ldots,n
\end{align}
are linearly independent in $\ell(X)$. Now, let $c_1,\ldots,c_m \in \bbr$ be such that
\begin{align*}
\sum_{i=1}^m c_i ( \exp(h_i) - 1 ) = 0.
\end{align*}
Then we have
\begin{align*}
- \bigg( \sum_{i=1}^m c_i \bigg) \exp(h_0) + \sum_{i=1}^m c_i \exp(h_i) = 0,
\end{align*}
which implies $c_1 = \ldots = c_m = 0$ by the linear independence of the functions (\ref{lin-ind-line}). Since $m \in \bbn$ was arbitrary, we deduce that the subspace $V$ is infinite dimensional.
\end{enumerate}
\end{proof}

\begin{proposition}\label{prop-lin-ind-exp-ps}
Let $f \in \ell(X)$ be such that there is a sequence $(x_n)_{n \in \bbn} \subset X$ with $f(x_n) \neq 0$ for all $n \in \bbn$ and $f(x_n) \to 0$ as $n \to \infty$. Then, for all $m \in \bbn$ and all pairwise different $t_1,\ldots,t_m \in \bbr$ the functions
\begin{align}\label{lin-ind-space}
\exp(t_j f), \quad j=1,\ldots,m
\end{align}
are linearly independent in $\ell(X)$.
\end{proposition}

\begin{proof}
This proof has some natural similarity to that in \cite[Thm.~7.1]{Tappe-Levy}. In the sequel, let $c_1,\ldots,c_m \in \bbr$ be such that
\begin{align*}
\sum_{j=1}^m c_j \exp(t_j f) = 0.
\end{align*}
We define the function
\begin{align*}
g : \bbr \to \bbr, \quad g(z) := \sum_{j=1}^m c_j \exp(t_j z).
\end{align*}
Then $g$ has the power series representation
\begin{align*}
g(z) = \sum_{i=0}^{\infty} a_i z^i, \quad z \in \bbr,
\end{align*}
where the coefficients are given by
\begin{align*}
a_i = \sum_{j=1}^m c_j t_j^i, \quad i \in \bbn_0.
\end{align*}
Indeed, by the power series representation of the exponential function, for each $z \in \bbr$ we have
\begin{align*}
g(z) = \sum_{j=1}^m c_j \exp(t_j z) = \sum_{j=1}^m c_j \sum_{i=0}^{\infty} \frac{(t_j z)^i}{i!} = \sum_{i=0}^{\infty} \frac{1}{i!} \bigg( \sum_{j=1}^m c_j t_j^i \bigg) z^i = \sum_{i=0}^{\infty} a_i z^i.
\end{align*}
We define the sequence $(z_n)_{n \in \bbn} \subset \bbr \setminus \{ 0 \}$ as $z_n := f(x_n)$ for $n \in \bbn$. Then we have $z_n \to 0$ as $n \to \infty$ and $g(z_n) = 0$ for all $n \in \bbn$. The identity theorem for power series applies and yields $a_i = 0$ for all $i \in \bbn_0$. Defining the Vandermonde matrix $A \in \bbr^{m \times m}$ as $A_{ji} := t_j^i$ for $i = 0,\ldots,m-1$ and $j=1,\ldots,m$, and the vector $c := (c_1,\ldots,c_m)^{\top} \in \bbr^m$, we obtain $A^{\top} \cdot c = 0$. Since $t_1,\ldots,t_m$ are pairwise different by assumption, we deduce that $c = 0$, which proves the linear independence of the functions (\ref{lin-ind-space}).
\end{proof}

\end{appendix}

\end{document}